\documentclass[a4paper,UKenglish,cleveref, autoref, thm-restate]{lipics-v2021}

\title{Parallel algorithm for pattern matching problems under substring consistent equivalence relations} 

\titlerunning{Parallel algorithm for pattern matching problems under SCERs} 

\author{Davaajav~Jargalsaikhan}{Graduate School of Information Sciences, Tohoku University, Sendai, Japan}{davaajav\_jargalsaikhan@shino.ecei.tohoku.ac.jp}{}{}
\author{Diptarama~Hendrian}{Graduate School of Information Sciences, Tohoku University, Sendai, Japan}{diptarama@tohoku.ac.jp}{https://orcid.org/0000-0002-8168-7312}{JSPS KAKENHI Grant Number JP19K20208}
\author{Ryo~Yoshinaka}{Graduate School of Information Sciences, Tohoku University, Sendai, Japan}{ryoshinaka@tohoku.ac.jp}{https://orcid.org/0000-0002-5175-465X}{JSPS KAKENHI Grant Numbers JP18K11150 and JP20H05703}
\author{Ayumi~Shinohara}{Graduate School of Information Sciences, Tohoku University, Sendai, Japan}{ayumis@tohoku.ac.jp}{https://orcid.org/0000-0002-4978-8316}{JSPS KAKENHI Grant Number JP21K11745}

\authorrunning{D. Jargalsaikhan et al.}

\Copyright{Davaajav Jargalsaikhan et al.}



\category{} 

\relatedversion{appears in the proceedings of the 33rd Annual Symposium on Combinatorial Pattern Matching (CPM 2022) \cite{JargalsaikhanHYS22}}

\nolinenumbers 

\hideLIPIcs

\usepackage[ruled,linesnumbered,algo2e,vlined]{algorithm2e}
\usepackage{standalone}
\usepackage{makecell}
\usepackage{mathtools}
\usepackage{tikz}
\usetikzlibrary{calc,positioning,fit,patterns,decorations.pathreplacing,calligraphy}

\DeclarePairedDelimiter{\ceil}{\lceil}{\rceil}
\DeclarePairedDelimiter{\floor}{\lfloor}{\rfloor}

\newcommand{\hati}{\hat{\imath}}
\newcommand{\hatj}{\hat{\jmath}}

\newcommand{\mrm}[1]{\mathrm{#1}}
\newcommand{\mcal}[1]{\mathcal{#1}}
\newcommand{\True}{\mathit{True}}
\newcommand{\False}{\mathit{False}}
\newcommand{\LCP}{\mathit{LCP}}
\newcommand{\Head}{\mathit{Head}}
\newcommand{\Tail}{\mathit{Tail}}
\newcommand{\tail}{\mathit{tail}}
\newcommand{\oldtail}{\mathit{old\_tail}}
\newcommand{\prev}[1]{\mathit{prev}_{#1}}
\newcommand{\rem}{\mathit{s}}

\SetKwProg{Fn}{Function}{}{end}\SetKwFunction{Dueling}{Dueling}%
\SetKwProg{Fn}{Function}{}{end}\SetKwFunction{PreprocessingSerial}{PreprocessingSerial}%
\SetKwProg{Fn}{Function}{}{end}\SetKwFunction{DuelingStageSerial}{DuelingStageSerial}%
\SetKwProg{Fn}{Function}{}{end}\SetKwFunction{SweepingStageSerial}{SweepingStageSerial}%
\SetKwProg{Fn}{Function}{}{end}\SetKwFunction{CheckParallel}
{GetTightMismatchPos}%
\SetKwProg{Fn}{Function}{}{end}\SetKwFunction{PreprocessingParallelAperiodic}
{PreprocessingParallelAperiodic}%
\SetKwProg{Fn}{Function}{}{end}\SetKwFunction{PreprocessingParallel}{PreprocessingParallel}%
\SetKwProg{Fn}{Function}{}{end}\SetKwFunction{DuelingStageParallel}{DuelingStageParallel}%
\SetKwProg{Fn}{Function}{}{end}\SetKwFunction{SweepingStageParallel}{SweepingStageParallel}%
\SetKwProg{Fn}{Function}{}{end}\SetKwFunction{FinalizeTail}
{FinalizeTail}
\SetKwProg{Fn}{Function}{}{end}\SetKwFunction{SatisfyHeadSparsity}
{SatisfySparsity}
\SetKwProg{Fn}{Function}{}{end}\SetKwFunction{Finalize}
{Finalize}
\SetKwProg{Fn}{Function}{}{end}\SetKwFunction{Merge}
{Merge}
\SetKwProg{Fn}{Function}{}{end}\SetKwFunction{GetZeros}
{GetZeros}
\SetAlgoLongEnd
\SetArgSty{textup}

\begin{document}

\maketitle

\begin{abstract}
    Given a text and a pattern over an alphabet,
    the pattern matching problem searches for all occurrences of the pattern in the text.
    An equivalence relation $\approx$ is called a substring consistent equivalence relation (SCER),
    if for two strings $X$ and $Y$, $X \approx Y$ implies $|X| = |Y|$ and $X[i:j] \approx Y[i:j]$ for all $1 \le i \le j \le |X|$.
    In this paper, we propose an efficient parallel algorithm for pattern matching under any SCER using the ``duel-and-sweep'' paradigm.
    For a pattern of length $m$ and a text of length $n$,
    our algorithm runs in $O(\xi_m^\mrm{t} \log^2 m)$ time and $O(\xi_m^\mrm{w} \cdot n \log^2 m)$ work,
    with $O(\tau_n^\mrm{t} + \xi_m^\mrm{t} \log^2 m)$ time and $O(\tau_n^\mrm{w} + \xi_m^\mrm{w} \cdot m \log^2 m)$ work preprocessing
    on the Priority Concurrent Read Concurrent Write Parallel Random-Access Machines (P-CRCW PRAM),
	where $\tau_n^\mrm{t}$, $\tau_n^\mrm{w}$, $\xi_m^\mrm{t}$, and $\xi_m^\mrm{w}$ are parameters dependent on SCERs, 
	which are often linearly bounded in $n$ and $m$, respectively.
\end{abstract}

\section{Introduction}

The string matching problem is fundamental and widely studied in computer science. 
Given a text and a pattern, the string matching problem searches for all substrings of the text that match the pattern.
Many matching functions that are used in different string matching problems, including exact~\cite{knuth1977fast}, parameterized~\cite{baker1996parameterized}, order-preserving~\cite{kim2014order, kubica2013linear} and cartesian-tree~\cite{park2019cartesian} matchings,
fall under the class of \emph{substring consistent equivalence relations} (SCERs)~\cite{matsuoka2016generalized}.
An equivalence relation on strings is an SCER, if two strings $X$ and $Y$ match under the equivalence relation, then they have equal length and
$X[i:j]$ matches $Y[i:j]$, for all $1 \leq i \leq j < |X|$.
Matsuoka et al.~\cite{matsuoka2016generalized} generalized the KMP algorithm~\cite{knuth1977fast} for pattern matching problems under SCERs.
They also investigated periodicity properties of strings under SCERs.
Kikuchi et al.~\cite{kikuchi2020computing} proposed algorithms to compute the shortest and longest cover arrays for a given string under any SCER.
Hendrian~\cite{hendrian2020generalized} generalized Aho-Corasick algorithm for the dictionary matching under SCERs.

Vishkin proposed two algorithms for exact pattern matching,
pattern matching by dueling~\cite{vishkin1985optimal}
and pattern matching by sampling~\cite{vishkin1991deterministic}.
Both algorithms match the pattern to a substring of the text from some positions which are determined by the property of the pattern,
instead of its prefix or suffix as in, for instance, the KMP algorithm~\cite{knuth1977fast}.
These algorithms are developed for parallel processing.

The dueling technique by Vishkin~\cite{vishkin1985optimal} has been proved to be useful for various kinds of pattern matching.
Amir et al.~\cite{amir1994alphabet} proposed a duel-and-sweep algorithm for two-dimensional exact matching, which is named ``consistency and verification''.
Cole et al.~\cite{cole2014two} extended it to two-dimensional parameterized matching.
In addition, Jargalsaikhan et al.~\cite{jargalsaikhan2018duel,jargalsaikhan2020parallel} proposed serial and parallel duel-and-sweep algorithms for order-preserving matching.

In this paper, we propose an efficient parallel algorithm based on the dueling technique for the pattern matching problem under SCERs.
Our parallel algorithm is the first to solve the problem under an arbitrary SCER in parallel. 
While Vishkin's dueling algorithm for exact matching depends on the preferable properties of periods of strings, many of those do not hold with SCERs.
Therefore, our algorithm involves new ideas and appears quite different from the original for exact pattern matching. 
For a pattern of length $m$ and a text of length $n$,
our algorithm runs in $O(\xi_m^\mrm{t} \log^2 m)$ time and $O(\xi_m^\mrm{w} \cdot n \log^2 m)$ work,
with $O(\tau_n^\mrm{t} + \xi_m^\mrm{t} \log^2 m)$ time and $O(\tau_n^\mrm{w} + \xi_m^\mrm{w} \cdot m \log^2 m)$ work preprocessing
on the Priority Concurrent Read Concurrent Write Parallel Random-Access Machines (P-CRCW PRAM)~\cite{jaja1992introduction}.
Here, $\tau_n^\mrm{t}$ and $\tau_n^\mrm{w}$ are time and work respectively, needed on P-CRCW PRAM to encode in parallel a string $X$ of length $n$ under the SCER in concern.
Given the encoding of $X$, $\xi_m^\mrm{t}$ and $\xi_m^\mrm{w}$ are time and work respectively to re-encode an element w.r.t.\ some suffix of $X$ of length $m$.
Table~\ref{table:contributions} shows the encoding time and work complexities for some SCERs.
\begin{table}[tb]
	\begin{center}
		\caption{Summary of the encoding complexities for some SCER on P-CRCW PRAM.}
		\label{table:contributions}
		\begin{tabular}{|l|c|c|c|c|}
			\hline
                  & $\tau_n^\mrm{t}$ & $\tau_n^\mrm{w}$ & $\xi_m^\mrm{t}$ & $\xi_m^\mrm{w}$ \\
            \hline \hline
            Exact & $O(1)$ & $O(1)$ & $O(1)$ & $O(1)$ \\
            \hline
            Parametererized & $O(\log n)$ & $O(n \log n)$ & $O(1)$ & $O(1)$\\
            \hline
            Cartesian-tree & $O(\log n)$ & $O(n \log n)$ & $O(1)$ & $O(1)$ \\ 
            \hline
		\end{tabular}
	\end{center}
\end{table}

This manuscript fixes minor errors and improves the algorithm efficiency in~\cite{JargalsaikhanHYS22}.

\section{Preliminaries}

We use $\Sigma$ to denote an alphabet of symbols and $\Sigma^*$ denotes the set of strings over the alphabet $\Sigma$.
For a string $X\in \Sigma^*$, the length of $X$ is denoted by $|X|$.
The \emph{empty string}, denoted by $\varepsilon$, is the string of length $0$.
For a string $X \in \Sigma^*$ of length $n$, $X[i]$ denotes the $i$-th symbol of $X$,
$X[i:j] = X[i]X[i+1] \dots X[j]$ denotes a substring of $X$ that begins at position $i$ and ends at position $j$ for $1 \leq i \leq j \leq n$.
For $i > j$, $X[i:j]$ denotes the empty string.

\begin{definition}[Substring consistent equivalence relation (SCER)~\cite{matsuoka2016generalized}]
	An equivalence relation $\approx \ \subseteq \Sigma^* \times  \Sigma^*$ is a \emph{substring consistent equivalence relation (SCER)} if 
	for two strings $X$ and $Y$, $X \approx Y$ implies $|X| = |Y|$ and $X[i:j] \approx Y[i:j]$ for all $1 \le i \le j \le |X|$.
\end{definition}
For instance, while the parameterized matching~\cite{baker1996parameterized} and
order-preserving matching~\cite{kubica2013linear, kim2014order} are SCERs,
the permutation matching~\cite{butman2004scaled, cicalese2009searching} and function matching~\cite{amir2006function} are not.

Hereafter we fix an arbitrary SCER $\approx$.
We say that a position $i$ is the \emph{tight mismatch position} if $X[1 \mathbin{:} i-1] \approx Y[1 \mathbin{:} i-1]$ and $X[1 \mathbin{:} i] \not\approx Y[1 \mathbin{:} i]$.
For two strings $X$ and $Y$, let $\LCP(X, Y)$ be the length $l$ of the longest prefixes of $X$ and $Y$ match.
That is, $l$ is the greatest integer such that $X[1 : l] \approx Y[1 : l]$.
Obviously, if $i$ is the tight mismatch position for $X \not\approx Y$, then $\LCP(X,Y) = i-1$.
The converse holds if $i \le \min\{|X|,|Y|\}$.
Similarly, for a string $X$ and an integer $0 \leq a < |X|$,
we define $\LCP_X(a) = \LCP(X,X[a+1:|X|])$.
In other words, $\LCP_X(a)$ is the length of the longest common prefix,
when $X$ is superimposed on itself with offset $a$. 
We say $X$ \emph{$\approx$-matches} $Y$ iff $X \approx Y$.
Given a text $T$ of length $n$ and a pattern $P$ of length $m$, a position $i$ in $T$,  $1 \le i \le n - m +1$, is an \emph{$\approx$-occurrence} of $P$ in $T$ iff $P \approx T[i:i+m-1]$.
\begin{definition}[$\approx$-pattern matching]
\mbox{}
	\begin{description}
	\item[Input:] A text $T \in \Sigma^*$ of length $n$ and a pattern $P \in \Sigma^*$ of length $m \le n$.
	\item[Output:] All $\approx$-occurrences of $P$ inside $T$.
	\end{description}
\end{definition}

In the remainder of this paper, we fix text $T$ to be of length $n$ and pattern $P$ to be of length $m$.
We also assume that $n = 2m - 1$.
Larger texts can be cut into overlapping pieces of length that are less than or equal to $(2m - 1)$ and processed independently.
That is, we search for pattern occurrences in each substring $T[1 \mathbin{:} 2m - 1], T[m+1 \mathbin{:} 3m-1], \dotsc, T[\floor{\frac{n-1}{m}} \cdot m + 1 \mathbin{:} n]$, independently.
For an integer $x$ with $1 \leq x \leq n - m +1$, 
a \emph{candidate} $T_x$ is the substring of $T$ starting from $x$ of length $m$, 
i.e., $T_x = T[x \mathbin{:} x+m-1]$.

For SCER matchings 
often it is convenient to encode the strings where $\approx$-equivalence is reduced to the identity.
Amir and Kondratovsky~\cite{amir2019sufficient} showed that every SCER admits an encoding satisfying the following property.\footnote{Lemma~12 in~\cite{amir2019sufficient} does not explicitly mention the third property, but their proof entails it.}
\begin{definition}[$\approx$-encoding]\label{def:encoding}
Let $\Sigma$ and $\Delta$ be alphabets.
We say a function $f:\Sigma^* \rightarrow \Delta^*$ is an \emph{$\approx$-encoding} if
\begin{enumerate}[(1)]
\item for any string $X \in \Sigma^*$, $|X| = |f(X)|$,
\item $f(X[1:i]) = f(X)[1:i]$ for any $i \le |X|$, 
\item for two strings $X$ and $Y$ of equal length $k$, $f(X)[i] = f(Y)[i]$ implies $f(X[j+1:k])[i-j] = f(Y[j+1:k])[i-j]$ for any $j < i \leq k$, and
\item $f(X) = f(Y)$ iff $X \approx Y$.
\end{enumerate}
\end{definition}
\begin{restatable}{proposition}{encoding}\label{prop:encoding}
An equivalence relation $\approx$ is an SCER if and only if it admits an $\approx$-encoding.
\end{restatable}
\begin{proof}
It suffices to show the ``if'' direction.
Suppose we have an $\approx$-encoding $f$.
If $X \approx Y$, then $f(X)=f(Y)$ by (4) of Definition~\ref{def:encoding}.
In this case, we have $f(X[j:k])[i]=f(Y[j:k])[i]$ for any $1 \le j \le k \le |X|$ and $1 \le i \le k-j+1$ by $f(X)[i+j-1]=f(Y)[i+j-1]$, (3), and (2).
Hence, $X[j:k] \approx Y[j:k]$ by (4).
\end{proof}
Standard encodings of SCERs often satisfy the above definition, such as the prev-encoding~\cite{baker1996parameterized} for parameterized matching and parent-distance encoding~\cite{park2019cartesian} for cartesian-tree matching.
However, the nearest neighbor encoding~\cite{kim2014order} for order-preserving matching violates the third condition. 
Our algorithm for $\approx$-pattern matching proposed in this paper relies on the property of Definition~\ref{def:encoding} and does not work with the nearest neighbor encoding.
Nonetheless, duel-and-sweep algorithms for order-preserving matching based on the encoding are possible by further elaboration~\cite{jargalsaikhan2018duel, jargalsaikhan2020parallel}, but we will not discuss it in this paper.

Fixing an $\approx$-encoding $f$, we denote $f(X)$ by $\widetilde{X}$ for simplicity.
In addition, we denote the encoding of $X[x:|X|]$ as $\widetilde{X}_{x} = f(X[x:|X|])$. Thus $\widetilde{X}_{1}=\widetilde{X}$.
For a string $X$, we suppose that $\widetilde{X}$ can be computed in $\tau_{|X|}^\mrm{t}$ time and $\tau_{|X|}^\mrm{w}$ work in parallel on P-CRCW PRAM.
Moreover, we assume that given $\widetilde{X}$, $x$, and $k$ such that $x+k-1 \le |X|$, to compute $\widetilde{X}_x[k]$, i.e.\ re-encoding the element at position $k$ with respect to suffix $X[x:|X|]$, takes $\xi_{k}^\mrm{t}$ time and $\xi_{k}^\mrm{w}$ work on P-CRCW PRAM.
Of course, one can obtain the value $\widetilde{X}_x[k]$ by compute the whole $\widetilde{X}_x[:k]$ in $\tau_{k}^\mrm{t}$ time and $\tau_{k}^\mrm{w}$ work, but re-encoding a single position is usually much cheaper.
Those parameters are often reasonably small.
See Table~\ref{table:contributions} and Appendix~\ref{app:sec:encoding} for the prev-encoding for parameterized matching and the parent-distance encoding for cartesian-tree matching.

Vishkin's dueling technique essentially depends on the preferable properties of periods of strings.
Matsuoka et al.~\cite{matsuoka2016generalized} have discussed in detail how the classical notion of periods and their properties can be generalized when considering SCER matching.
Unfortunately, none of the generalizations yield a straightforward adaptation of Vishkin's algorithm for SCER matching.
Among those, the kind of periods involved in the duel-and-sweep algorithm discussed in this paper is \emph{border-based period}.
\begin{definition}[Border-based period]\label{def:border-based}
	Given a string $X$ of length $n$, positive integer $p < n$ is called a \emph{border-based period} of $X$ if $X[1 \mathbin{:} n-p]  \approx X[p+1 \mathbin{:} n]$. 
\end{definition}
Throughout the rest of the paper, we will refer to a border-based period as a \emph{period}.

The family of models of computation used in this work is the priority concurrent-read concurrent-write (P-CRCW) PRAM~\cite{jaja1992introduction}.
This model allows simultaneous reading from the same memory location as well as simultaneous writing. 
In case of multiple writes to the same memory cell, the P-CRCW PRAM grants access to the memory cell to the processor with the smallest index.

\section{Parallel algorithm for pattern matching under SCERs}

We give an overview of the duel-and-sweep algorithm~\cite{amir1994alphabet,vishkin1985optimal}.
The pattern is first preprocessed to obtain a \emph{witness table}, which is later used to prune candidates during the pattern searching.
As the name suggests, in the duel-and-sweep algorithm, the pattern searching is divided into two stages: the \emph{dueling stage} and the \emph{sweeping stage}.
The pattern searching algorithm prunes candidates that cannot be pattern occurrences, first by performing ``duels'' between them, and then by ``sweeping'' through the remaining candidates to obtain pattern occurrences. 

First, we explain the idea of dueling.
Suppose $P$ is superimposed on itself with an offset $a < m$ and the two overlapped regions of $P$ do not match under $\approx$.
Then it is impossible for two candidates $T_x$ and $T_{x+a}$ with offset $a$ to match $P$ simultaneously (see Figure~\ref{fig:overview_dueling}).
The dueling stage lets each pair of candidates with such offset $a$ ``duel'' and eliminates one based on this observation,
so that if candidate $T_x$ gets eliminated during the dueling stage, then $T_x \not\approx P$.
However, the opposite does not necessarily hold true: $T_x$ surviving the dueling stage does not mean that $T_x \approx P$. 
On the other hand, it is guaranteed that if distinct candidates $T_x$ and $T_{x+a}$ that survive the dueling stage overlap, then the suffixes of $T_x$ and $P$ of length $m-a$ match if and only if so do the prefixes of $T_{x+a}$ and $P$ of the same length.
The sweeping stage takes advantage of this property when checking whether surviving candidates and the pattern match, so that this stage can also be done quickly.

Prior to the dueling stage, the pattern is preprocessed to construct a \emph{witness table} based on which the dueling stage decides which pair of overlapping candidates should duel and how they should duel.
For each offset $0 \leq a < m$,
when the overlapped regions obtained by superimposing $P$ on itself with offset $a$ do not match, we need only one position $i$ to say that
the overlapping regions do not match.
We say that $w$ is a \emph{witness for the offset $a$} if $\widetilde{P}_{a+1}[w] \neq \widetilde{P}[w]$.
We denote by $\mathcal{W}_P(a)$ the set of all witnesses for offset $a$.
We say a witness $w$ for offset $a$ is \emph{tight} if $w = \min \mathcal{W}_P(a)$.
Obviously, $\mathcal{W}_P(a) = \emptyset$ if and only if $a=0$ or $a$ is a period of $P$.
A \emph{witness table} $W[0 \mathbin{:} m-1]$ is an array such that $W[a] \in \mathcal{W}_P(a)$ if $\mathcal{W}_P(a) \neq \emptyset$.
When the overlap regions match for offset $a$, which implies that no witness exists for $a$,
we express it as $W[a] = 0$.

\begin{algorithm2e}[tb]
	\caption{Dueling with respect to $S$. There is one survivor assuming $x$ is not consistent with $y$.}
	\label{alg:dueling}
	\SetVlineSkip{0.5mm}
    \Fn{\Dueling{$\widetilde{S}, x, y$}}{
        $w \leftarrow W[y-x]$\;
		\lIf{$\widetilde{S}_y[w] = \widetilde{P}[w]$}{%
			\KwRet{$y$}%
		}
		\lElse{\KwRet{$x$}}
	}
\end{algorithm2e}

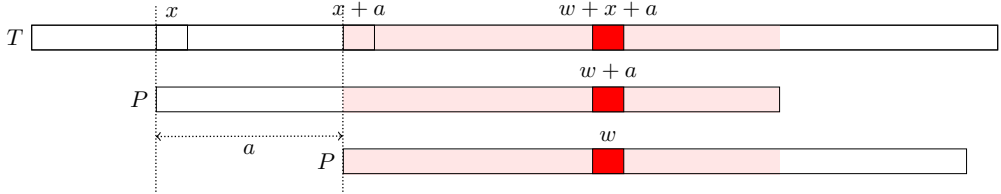
\begin{figure}[t]
	\centering
	\tikzset{
    charbox/.style={draw, rectangle, minimum height= 4mm, minimum width = 5mm,
    label={[align=center,font=\large]above:#1}},
    charboxbelow/.style={draw, rectangle, minimum height= 4mm, minimum width = 5mm,
    label={[align=center,font=\large]below:#1}},
    fullbox/.style={draw, rectangle, minimum height= 4mm, minimum width = 15.5cm,
    label={[align=right,font=\large]left:#1}},
    halfbox/.style={draw, rectangle, minimum height= 4mm, minimum width = 10cm,
    label={[align=right,font=\large]left:#1}},
    overlap/.style={fill=red!10, rectangle, minimum height= 4mm, minimum width = 7cm,
    label={[align=right,font=\large]left:#1}},
    blockline/.style={densely dotted, semithick},
    fittedrec/.style n args={2}{draw, fit={(#1)(#2)}, inner sep=0, outer sep=0, pattern=north east lines}
}

\begin{tikzpicture}
\node[fullbox={}] (T) {};
\node[overlap={},right = 5cm of T.west] (O) {};
\node[overlap={},below = 0.6cm of O.south] {};
\node[overlap={},below = 1.6cm of O.south] {};

\node[fullbox=$T$] (T0) {};
\node[charbox=$x$, right = 2cm of T.west] (x) {};
\node[charbox=$x+a$, right = 5cm of T.west] (xa) {};
\node[charbox=$w+x+a$, right = 9cm of T.west, fill=red] (wxa) {};

\node[halfbox=$P$, below = 1cm of x.west, anchor = west] (P1) {};
\node[charbox=$w+a$, below = 1cm of wxa.west, anchor = west, fill=red] {};

\node[halfbox=$P$, below = 2cm of xa.west, anchor = west] (P2) {};
\node[charbox=$w$, below = 2cm of wxa.west, anchor = west, fill=red] {};

\draw[blockline] ($(x.north west) + (0, 3mm)$) -- ($(x.south west) + (0, -2.3cm)$);
\draw[blockline] ($(xa.north west) + (0, 3mm)$) -- ($(xa.south west) + (0, -2.3cm)$);

\draw[blockline, <->] ($(P1.west) + (0, -0.6cm)$) -- ($(P2.north west)+(0,0.2cm)$) node[below, midway, font=\large]{$a$};

\end{tikzpicture}
	\caption{
		If $T_x \approx P \approx T_{x+a}$, then the overlapped regions of $P$ superimposed on itself with offset $a$ should match, i.e., $P_{a+1}[1:m-a]\approx P[1:m-a]$.
		If the overlapped region does not match, there must be a witness $w$ such that $\widetilde{P}_{a+1}[w] \neq \widetilde{P}[w]$.
		Candidate positions $x$ and $x+a$ perform a duel using the witness $w$ based on Lemma~\ref{lem:duel}.}
	\label{fig:overview_dueling}
\end{figure}

More formally, in the dueling stage, we ``duel'' positions $x$ and $x+a$ such that $\mathcal{W}_P(a) \neq \emptyset$ based on the following observation (see Figure~\ref{fig:overview_dueling}).
\begin{restatable}{lemma}{duel}\label{lem:duel}
Suppose $w \in \mathcal{W}_P(a)$. Then, 
\begin{itemize}
	\item if $\widetilde{T}_{x+a}[w] = \widetilde{P}[w]$, then $T_{x} \not\approx P$,
	\item if $\widetilde{T}_{x+a}[w] \neq \widetilde{P}[w]$, then $T_{x+a} \not\approx P$.
\end{itemize}
\end{restatable}
\begin{proof}
If $\widetilde{T}_{x+a}[w] \neq \widetilde{P}[w]$, then by the fourth property of the $\approx$-encoding (Definition~\ref{def:encoding}), $T_{x+a} \not\approx P$.
If $\widetilde{T}_{x+a}[w] = \widetilde{P}[w] \neq \widetilde{P}_{a+1}[w]$,
then by the third property of the $\approx$-encoding, $\widetilde{T}_{x}[w+a] \neq \widetilde{P}[w+a]$, so $T_x \not\approx P$.
\end{proof}

Based on this lemma, we can safely eliminate either candidate $T_x$ or $T_{x+a}$ without looking into other positions.
This process is called \emph{dueling} and described as Algorithm~\ref{alg:dueling}.
Since it compares just a single positon, it runs in $O(\xi^{\mrm{t}}_m)$ time and $O(\xi^{\mrm{w}}_m)$ work assuming that $\widetilde{S}$, $\widetilde{P}$, and $W$ have already been computed.
On the other hand, if the offset $a$ has no witness, i.e.\ $P[1 \mathbin{:} m-a] \approx P[a+1 \mathbin{:} m]$, no dueling is performed on them.
We say that a position \emph{$x$ is consistent with $x+a$} if $\mathcal{W}_P(a) = \emptyset$.


After the dueling stage, all surviving candidate positions are pairwise consistent.
The dueling stage algorithm makes sure that no occurrence gets eliminated during the dueling stage.
Taking advantage of the fact that surviving candidates from the dueling stage are pairwise consistent, the sweeping stage prunes them until all remaining candidates match the pattern.
By ensuring pairwise consistency of the surviving candidates, the pattern searching algorithm reduces the number of comparisons at a position in the text during the sweeping stage.

Hereinafter, in our pseudo-codes we will use ``$\leftarrow$'' to note assignment operation into a local variable of a processor or assignment operation into a global variable which is
accessed by a single processor at a time.
We will use ``$\Leftarrow$'' to note assignment operation into a global variable which is accessible from multiple processors simultaneously.
In case of a write
conflict, the processor with the smallest index succeeds in writing into the memory.

\subsection{Pattern preprocessing}

The goal of the preprocessing stage is to compute a witness table
$W[0 \mathbin{:} m-1]$, where $W[a] = 0$ if $\mathcal{W}_P(a) = \emptyset$, and
$W[a] \in \mathcal{W}_P(a)$ otherwise.
Algorithm~\ref{alg:check_SCER_parallel} computes the tight mismatch position for $X$ and $Y$, given $\widetilde{X}$ and $\widetilde{Y}$.
\begin{restatable}{lemma}{checkparallel}\label{lem:check_parallel}
	For strings $X$ and $Y$ of equal length, given $\widetilde{X}$ and $\widetilde{Y}$, Algorithm~\ref{alg:check_SCER_parallel} computes the tight mismatch position in $O(1)$ time and $O(|X|)$ work on the P-CRCW PRAM.
\end{restatable}
\begin{proof}
	In Algorithm~\ref{alg:check_SCER_parallel}, 
	for each element of $X$, we ``attach'' a processor to each position of $X$.
	If $\widetilde{X}[i] \neq \widetilde{Y}[i]$ for some $i$, the corresponding processor tries to update the shared variable $w$.
	Recall that in P-CRCW PRAM, the processor with the lowest index will succeed in writing into $w$.
	Thus, at the end of the algorithm $w$ contains the tight mismatch position.
\end{proof}

One can compute a witness table naively inputting $\widetilde{P}[1:m-a]$ and $\widetilde{P}_{a+1}[1:m-a]$ for all the offsets $a < m$ to Algorithm~\ref{alg:check_SCER_parallel}.
However, this naive method costs as much as $\Omega(\xi_m^\mrm{w} \cdot m^2)$ work.
We will present a more efficient algorithm in this subsection.

\begin{algorithm2e}[t]
	\caption{Check in parallel whether $X$ and $Y$ match, given $\widetilde{X}$ and $\widetilde{Y}$. If they do not match, it returns the tight witness.}
	\label{alg:check_SCER_parallel}
	\Fn(){\CheckParallel{$\widetilde{X}, \widetilde{Y}$}}{
		$w \leftarrow 0$\;
		\ForPar{\textbf{each} $i \in \{1, \dotsc, |X|\}$}{
			\lIf{$\widetilde{X}[i] \neq \widetilde{Y}[i]$}{%
				$w \Leftarrow i$%
			}
		}
		\KwRet{$w$}\;
	}
\end{algorithm2e}

Our pattern preprocessing algorithm is described in Algorithm~\ref{alg:preprocessing_parallel}
and its outline is illustrated in Figure~\ref{fig:preprocessing_invariant}.
Initially, all entries of the witness table are set to zero.
Throughout preprocessing, each element of $W$ is updated at most once.
Therefore, at any point of the execution of the preprocessing algorithm, if $W[i] \neq 0$, then it must hold $W[i] \in \mcal{W}_P(i)$.
We say that position $i$ is \emph{finalized} if $W[i] = 0$ implies $\mcal{W}_P(i) = \emptyset$ and $W[i] \neq 0$ implies $W[i] \in \mcal{W}_P(i)$.
During the execution of Algorithm~\ref{alg:preprocessing_parallel}, the table is divided into two parts.
The \emph{head} is a prefix of a certain length and the \emph{tail} is the rest suffix.
Let us write the head and the tail at the round $k$ of the while loop by $\Head_k$ and $\Tail_k$, respectively. 
The variable $\tail$ in Algorithm~\ref{alg:preprocessing_parallel} represents the starting position of the tail, or equivalently, the length of the head.
Throughout the algorithm execution, the tail part is always finalized.
On the other hand, though the zero entries of the head are not necessarily reliable, such zero positions become fewer and fewer.
Consider partitioning the head into blocks of size $2^k$. 
We will call each block a \emph{$2^k$-block}, with the last $2^k$-block possibly being shorter than $2^k$.
That is, the $2^k$-blocks are $W[i\cdot 2^k \mathbin{:} (i+1) \cdot 2^k-1]$ for $i=0,\dots, \floor{h/2^k} - 1$ and $W[\floor{h/2^k} \cdot 2^k \mathbin{:} h-1]$ where $h=|\Head_k|$ is the size of the head.
We say that $W[0 \mathbin{:} x]$ is \emph{$2^k$-sparse} if every $2^k$-block of $W[0:x]$ contains exactly one zero entry possibly except that the last $2^k$-block has no zero entry.
We will guarantee that $\Head_k$ is $2^k$-sparse.
Note that when the head is $2^k$-sparse, the unique zero position of the first $2^k$-block $W[0 \mathbin{:} 2^k - 1]$ is always $0$ ($W[0] = 0$) and $W[1 \mathbin{:} 2^k - 1]$ contains no zeros.
\begin{algorithm2e}[tb]
	\caption{Parallel algorithm for the pattern preprocessing.}
	\label{alg:preprocessing_parallel}
    \Fn(){\PreprocessingParallel{}}{
        $\tail \leftarrow m, \ k \leftarrow 0$\tcc*{$\tail$ is the starting position of $\Tail_k$}
        \While{$2^k \leq \tail$}{
            $p \leftarrow $ \GetZeros{$2^{k}, 2^{k+1} - 1, k$}$[0]$\;\label{alg:line:find_pk}
            $W[p] \leftarrow $ \CheckParallel{$\widetilde{P}[1 : m - p], \widetilde{P}_{p+1}[1 : m - p]$}\;\label{alg:line:finalize_pk}
            \lIf{$W[p] = 0$}{%
                $\mathit{lcp} \leftarrow m - p$%
            }\lElse{%
                $\mathit{lcp} \leftarrow W[p] - 1$%
            }
            $\oldtail \leftarrow \tail$\;
            $\tail \leftarrow \min(\oldtail -  2^k, m - \mathit{lcp})$\;
            \SatisfyHeadSparsity{$\tail-1, k$}\;\label{alg:line:spss}
            \FinalizeTail{$\tail, \oldtail, p,k$}\;
            $k \leftarrow k+1$\;
        }
    }
\end{algorithm2e}

Initially, the entire table is the head and the size of the tail is zero: $\Head_0 = W$ and $\Tail_0 = \varepsilon$.
The head is shrunk and the tail is extended by the following rule.
Let the \emph{suspected period} $p_k$ at round $k$ be the first zero position after the index $0$, i.e., $p_k$ is the unique position in the second $2^k$-block such that $W[p_k]=0$.
Then, we let $\mathit{Head}_{k+1} = W[0 : m-x-1]$ and $\Tail_{k+1} = W[m - x : m-1]$ for $x = |\Tail_{k+1}| =  \max(|Tail_{k}| + 2^{k},\ \LCP_P(p_k))$.
When $|\Head_k| < 2^k$, the $2^k$-sparsity means that all the positions in the witness table are finalized.
So, Algorithm~\ref{alg:preprocessing_parallel} exits the while loop and halts.
The goal of this subsection is to show the following theorem.
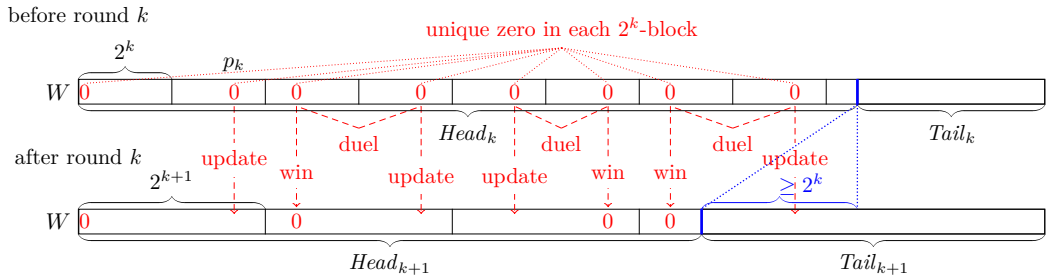
\begin{figure}[t]
	\centering
	\tikzset{
    charbox/.style={draw, rectangle, minimum height= 4mm, minimum width = 5mm,
    label={[align=center,font=\large]above:#1}},
    charboxbelow/.style={draw, rectangle, minimum height= 4mm, minimum width = 5mm,
    label={[align=center,font=\large]below:#1}},
    fullbox/.style={draw, rectangle, minimum height= 4mm, minimum width = 15.5cm,
    label={[align=right,font=\large]left:#1}},
    halfbox/.style={draw, rectangle, minimum height= 4mm, minimum width = 10cm,
    label={[align=right,font=\large]left:#1}},
    blockline/.style={densely dotted, semithick},
    fittedrec/.style n args={2}{draw, fit={(#1)(#2)}, inner sep=0, outer sep=0, pattern=north east lines}
}

\begin{tikzpicture}
\node[fullbox=$W$] (W1) {};
\node[fullbox=$W$, below = 2.1cm of W1.west, anchor = west] (W2) {};

\foreach \i in {0,...,8} {
    \draw[black] ($(W1.north west) + (\i * 15mm, 0mm)$) -- ($(W1.south west) + (\i * 15mm, 0mm)$);
}
\foreach \i in {0,...,3} {
    \draw[black] ($(W2.north west) + (\i * 30mm, 0mm)$) -- ($(W2.south west) + (\i * 30mm, 0mm)$);
}


\node[right = 2.5cm of W1.north west, anchor = south] () {$p_k$};
\node[right = 0.1cm of W1.west, anchor = center,red] (z10) {$0$};
\node[right = 2.5cm of W1.west, anchor = center,red] (z11) {$0$};
\node[right = 3.5cm of W1.west, anchor = center,red] (z12) {$0$};
\node[right = 5.5cm of W1.west, anchor = center,red] (z13) {$0$};
\node[right = 7cm of W1.west, anchor = center,red] (z14) {$0$};
\node[right = 8.5cm of W1.west, anchor = center,red] (z15) {$0$};
\node[right = 9.5cm of W1.west, anchor = center,red] (z16) {$0$};
\node[right = 11.5cm of W1.west, anchor = center,red] (z17) {$0$};
\draw[pen colour={black}, decorate,
    decoration = {calligraphic brace,amplitude=2mm}] (W1.north west) --  ($(W1.north west)+(1.5cm,0cm)$) node[above=2mm,midway]{$2^k$};
    
\node[right = 2.5cm of W2.north west, anchor = south] (p2) {};
\node[right = 0.1cm of W2.west, anchor = center,red] (z20) {$0$};
\node[right = 2.5cm of W2.west, anchor = center,red] (z21) {};
\node[right = 3.5cm of W2.west, anchor = center,red] (z22) {$0$};
\node[right = 5.5cm of W2.west, anchor = center,red] (z23) {};
\node[right = 7cm of W2.west, anchor = center,red] (z24) {};
\node[right = 8.5cm of W2.west, anchor = center,red] (z25) {$0$};
\node[right = 9.5cm of W2.west, anchor = center,red] (z26) {$0$};
\node[right = 11.5cm of W2.west, anchor = center,red] (z27) {};
\draw[pen colour={black}, decorate,
    decoration = {calligraphic brace,amplitude=2mm}] (W2.north west) --  ($(W2.north west)+(3cm,0cm)$) node[above=2mm,midway]{$2^{k+1}$};

\draw[red,densely dashed,->] (z11.south) to node[minimum height=0.5cm,fill=white] {update}  (z21.north);
\node[red] (d1) at ($(z12.south)!.5!(z13.south)-(0cm,0.6cm)$) {duel};
\node[red] (d2) at ($(z14.south)!.5!(z15.south)-(0cm,0.6cm)$) {duel};
\node[red] (d3) at ($(z16.south)!.5!(z17.south)-(0cm,0.6cm)$) {duel};
\draw[red,densely dashed] (d1.north) -- (z12.south);
\draw[red,densely dashed] (d1.north) -- (z13.south);
\draw[red,densely dashed] (d2.north) -- (z14.south);
\draw[red,densely dashed] (d2.north) -- (z15.south);
\draw[red,densely dashed] (d3.north) -- (z16.south);
\draw[red,densely dashed] (d3.north) -- (z17.south);
\draw[red,densely dashed,->] (z13.south) to node[minimum height=0.5cm,fill=white,below=0cm] {update} (z23.north);
\draw[red,densely dashed,->] (z14.south) to node[minimum height=0.5cm,fill=white,below=0cm] {update} (z24.north);
\draw[red,densely dashed,->] (z17.south) to node[minimum height=0.5cm,fill=white] {update} (z27.north);
\draw[red,densely dashed,->] (z12.south) to node[minimum height=0.5cm,fill=white,below=0cm] {win} (z22.north);
\draw[red,densely dashed,->] (z15.south) to node[minimum height=0.5cm,fill=white,below=0cm] {win} (z25.north);
\draw[red,densely dashed,->] (z16.south) to node[minimum height=0.5cm,fill=white,below=0cm] {win} (z26.north);


\node[draw, rectangle, minimum height= 4mm, minimum width = 3cm, black, left = 0cm of W1.east] (Tail1) {};
\node[draw, rectangle, minimum height= 4mm, minimum width = 5.5cm, black, left = 0cm of W2.east] (Tail2) {};

\draw [pen colour={black}, decorate,
    decoration = {calligraphic brace,mirror,amplitude=2mm}] (Tail1.south west) --  (Tail1.south east) node[below=2mm,midway]{$\mathit{Tail}_k$};
\draw [pen colour={black}, decorate,
    decoration = {calligraphic brace,mirror,amplitude=2mm}] (Tail2.south west) --  (Tail2.south east) node[below=2mm,midway]{$\mathit{Tail}_{k+1}$};
\draw [decorate,
    decoration = {calligraphic brace,mirror,amplitude=2mm}] (W1.south west) --  (Tail1.south west) node[below=2mm,midway]{$\mathit{Head}_{k}$};
\draw [decorate,
    decoration = {calligraphic brace,mirror,amplitude=2mm}] (W2.south west) --  (Tail2.south west) node[below=2mm,midway]{$\mathit{Head}_{k+1}$};

\draw[very thick, blue] (Tail1.north west) -- (Tail1.south west);
\draw[very thick, blue] (Tail2.north west) -- (Tail2.south west);
\draw[blockline, blue] (Tail1.south west) -- ($(Tail2.north west) + (2.5cm, 0cm)$);
\draw[blockline, blue] (Tail1.south west) -- (Tail2.north west);

\draw [pen colour={blue}, decorate,
    decoration = {calligraphic brace, amplitude=2mm}] (Tail2.north west) -- ($(Tail2.north west) + (2.5cm, 0cm)$) node[above=1mm,midway,blue]{\quad\quad$\geq 2^k$};

\node[above = 10mm of W1.west] {before round $k$};
\node[above = 8mm of W2.west] {after round $k$};
\node[above = 5mm of W1.north,red] (uz) {unique zero in each $2^k$-block};
\foreach \i in {0,...,7} {
\draw[densely dotted,red] (uz.south) -- ($(z1\i.north)-(0cm,0.1cm)$);
}
\end{tikzpicture}
	\caption{Illustration of the preprocessing invariant. $W$ is partitioned into head and tail.
	The head is $2^k$-sparse and the tail is finalized.
	The $2^k$-sparsity is achieved by duels.
	The tail grows by at least $2^k$ at each round.
	\label{fig:preprocessing_invariant}}
\end{figure}

\begin{restatable}{theorem}{preprocessingcomplexity}
\label{th:preprocessing_complexity}
	Given $\widetilde{P}$, the pattern preprocessing Algorithm~\ref{alg:preprocessing_parallel} computes a witness table in $O(\xi_m^\mrm{t} \cdot \log^2 m)$ time and $O(\xi_m^\mrm{w} \cdot m \log^2 m)$ work on the P-CRCW PRAM.
\end{restatable}
\begin{proof}
By Lemmas~\ref{lem:satisfy_sparsity} and~\ref{lem:extend_tail}.
\end{proof}

\begin{algorithm2e}[t]
	\caption{Assuming that $W[0 \mathbin{:} r]$ is $2^k$-sparse, returns positions of zeros in $W[l:r]$.}
	\label{alg:get_zeros}
	\Fn(){\GetZeros{$l, r, k$}}{
		\textbf{create} array $A[0:\floor{r/2^k} - \floor{l/2^k}]$ and initialize elements to $-1$\;
		\ForPar{\textbf{each} $i \in \{l, \dotsc, r\}$}{%
            \lIf{$W[i] = 0$}{
				$A[\floor{i/2^k} - \floor{l/2^k}] \Leftarrow i$%
			}
        }
		\KwRet{$A$}\;
	}
\end{algorithm2e}

In the remainder of this subsection, we explain how to maintain the $2^k$-sparsity of the head and finalize the tail.
Before going into the detail, we prepare a technical function \GetZeros{$l,r,k$} in Algorithm~\ref{alg:get_zeros}, which returns positions $i \in \{l, \dotsc, r\}$ such that $W[i] = 0$ in an array, assuming that $W[0 \mathbin{:} r]$ satisfies the $2^k$-sparsity.
Algorithm~\ref{alg:get_zeros} runs in $O(1)$ time and $O(r)$ work on the P-CRCW PRAM.

\paragraph*{Comparison with Vishkin's algorithm}
The preprocessing algorithm for exact matching by Vishkin~\cite{vishkin1985optimal} also constructs a witness table so that it satisfies the $2^k$-sparsity, incrementing $k$, where it has no head/tail separation.
Maintaining the $2^k$-sparsity for the whole table is possible due to the periodicity property which holds for the exact identity but not for general SCERs.
Let $p \le \floor{i/2}$ be the shortest period of $P[:i]$ for some $i$.
In exact matching $P[i] \ne P[i+j-1]$ implies $P[i-p] \ne P[i+j-1]$.
Thus, we can update $W[j+p]$ by using $W[j]$, i.e., we may let $W[j+p] = W[j] - p$.
However, this property does not hold on SCERs generally.
Still, Vishkin's technique for keeping the $2^k$-sparsity can partially be applied to SCER cases under a certain condition (Lemma~\ref{lem:headwitnesses}).
To fulfill the condition, we control the length of the head part carefully.
Concerning the tail part, where Vishkin's technique does not work, we design a new efficient algorithm for computing witnesses.

\subsubsection*{Head invariant}
First we discuss how the algorithm makes $\Head_k$ $2^k$-sparse.
We maintain the head so that at the beginning of round $k$ of Algorithm~\ref{alg:preprocessing_parallel}, it satisfies the following invariant properties.
\begin{itemize}
	\item $\Head_k$ is $2^k$-sparse.
	\item For all positions $i$ of $\Head_k$,
	\begin{itemize}
	    \item $W[i] \neq 0$ implies $W[i] \in \mcal{W}_P(i)$,
	    \item $W[i] \leq |\Tail_k| + 2^k$.
    \end{itemize}
\end{itemize}
\begin{algorithm2e}[tb]
	\caption{Satisfy $2^{k+1}$-sparsity of $\Head_{k+1} = W[0 \mathbin{:} x]$.}
	\label{alg:satisfy_sparsity}
	\SetVlineSkip{0.5mm}
	\Fn{\SatisfyHeadSparsity{$x, k$}}{
        $A \leftarrow $ \GetZeros {$2^{k+1}, x, k$}\;\label{alg:line:ssgz}
        \ForPar{\textbf{each} $i \in \{0,1, \dotsc, \floor{|A|/2-1}\}$}{
            $j_1 \leftarrow A[2i],\ j_2 \leftarrow A[2i+1]$\;
            \If{$j_1 \neq -1$ and $j_2 \neq -1$}{
                $surv \leftarrow $ \Dueling{$\widetilde{P}, j_1, j_2$}\;
                $a \leftarrow j_2 - j_1$\;
                \lIf{$surv = j_1$}{%
                    $W[j_2] \Leftarrow W[a]$%
                }
                \lIf{$surv = j_2$}{%
                    $W[j_1] \Leftarrow W[a] + a$%
                }
            }
        }
	}
\end{algorithm2e}

The head maintenance procedure \SatisfyHeadSparsity is described in Algorithm~\ref{alg:satisfy_sparsity}.
Before calling the function \SatisfyHeadSparsity, Algorithm~\ref{alg:preprocessing_parallel} finalizes the suspected period $p_k$, the first position after 0 such that $W[p_k]=0$.
Due to the $2^k$-sparcity, $2^{k} \leq p_k < 2^{k+1}$.
Algorithm~\ref{alg:preprocessing_parallel} finds the suspected period $p_k$ at Line~\ref{alg:line:find_pk}
and then finalizes the position $p_k$ at Line~\ref{alg:line:finalize_pk}.

Let us explain how Algorithm~\ref{alg:satisfy_sparsity} works.
The task of $\SatisfyHeadSparsity(x,k)$ is to make $W[0:x]$ satisfy the $2^{k+1}$-sparsity.
In the case where the suspected period $p_k$ is the smallest period of $P$, i.e., $\mcal{W}_P(p_k)=\emptyset$, we have $\tail = m-\LCP_P(p_k)=p_k<2^{k+1}$ when Algorithm~\ref{alg:preprocessing_parallel} calls \SatisfyHeadSparsity{$\tail-1,k$}.
Then the array $A$ obtained at Line~\ref{alg:line:ssgz} is empty and \SatisfyHeadSparsity{$\tail-1,k$} does nothing.
After finalizing $\Tail_{k+1}$, which will be explained later, the algorithm will halt without going into the next loop, since $|\Head_{k+1}| \leq m-\LCP_P(p_k) = p_k < 2^{k+1}$.
At that moment all positions of $W$ are finalized.

Hereafter we suppose that $p_k$ is not a period of $P$.
When \SatisfyHeadSparsity{$\tail-1,k$} is called, the value of $W[p_k]$ is the tight witness and the first $2^{k+1}$-block contains no zeros except $W[0]$.
At that moment, the other part of the head is $2^{k}$-sparse.
To make it $2^{k+1}$-sparse, we perform duels between two zero positions $i$ and $j$ ($i < j$) within each of the $2^{k+1}$-blocks of the head except for the first one. 
The witness used for the duel between $i$ and $j$ is $W[a]$ for $a=j-i$, which is in the first $2^{k+1}$-block.
The following two lemmas ensure that indeed such duels are possible.
Suppose that the pattern is superimposed on itself with offsets $i$ and $j$.
Lemma~\ref{lem:headwitnesses} below claims that if we already know $w \in \mcal{W}_P(a)$ and $j+w \leq m$, in other words, if the witness lies within the overlap region, then we can obtain a witness for one of the offsets $i$ and $j$ by dueling them using $w$, without looking into other positions.
Lemma~\ref{lem:witness_limit} ensures that indeed we have a witness $w = W[a]$ in our table such that $j+w \leq m$ holds, thanks to the invariant property.
\begin{restatable}{lemma}{headwitness}\label{lem:headwitnesses}
For two offsets $i$ and $j=i+a$ with $a > 0$, suppose $w \in \mcal{W}_P(a)$ and $j+w \leq m$. Then,
\begin{enumerate}
	\item if the offset $j$ survives the duel, i.e., $\widetilde{P}_{j+1}[w] = \widetilde{P}[w]$, then $w + a \in \mathcal{W}_P(i)$;
	\item if the offset $i$ survives the duel, i.e., $\widetilde{P}_{j+1}[w] \neq \widetilde{P}[w]$, then $w \in \mathcal{W}_P(j)$.
\end{enumerate}
\end{restatable} 
\begin{proof}
	If $\widetilde{P}_{j+1}[w] \neq \widetilde{P}_{1}[w]$, then $w \in \mathcal{W}_P(j)$ by definition.
	Suppose $\widetilde{P}_{j+1}[w] = \widetilde{P}_{1}[w]$.
	The fact $w \in \mathcal{W}_P(a)$ means
	\(
	\widetilde{P}_{1}[w] \neq \widetilde{P}_{a+1}[w]
	\) and thus
	$ \widetilde{P}_{j+1}[w] \neq \widetilde{P}_{a+1}[w]$.
	By Property (3) of the $\approx$-encoding (Definition~\ref{def:encoding}), we have
	\(
	\widetilde{P}_{i+1}[w+a] \neq \widetilde{P}_{1}[w+a],
	\)
	which means $w+a \in \mathcal{W}_P(i)$.
\end{proof}

\begin{restatable}{lemma}{inrange}
\label{lem:witness_limit}
    For round $k$, suppose the preprocessing invariant holds true and $\mcal{W}_P(p_k) \neq \emptyset$.
	Then, when \SatisfyHeadSparsity is about to be called at Line~\ref{alg:line:spss} of Algorithm~\ref{alg:preprocessing_parallel},
	for any two positions $i$ and $j$ of $\Head_{k+1}$ such that $0 < j -i < 2^{k+1}$, it holds that $j + W[j-i] \leq m$.
\end{restatable}
\begin{proof}
	Let $a = j -i$ and $w = W[a]$.
    Recall that $a$ belongs to the first $2^{k+1}$-block and $W[a]$ is updated only if $a = p_k$.
    Suppose $a \neq p_k$.
	At the beginning of round $k$, by the invariant property, we have $w \leq |\Tail_{k}| + 2^{k}$.
	Since $j < |\mathit{Head}_{k+1}| = m - |Tail_{k+1}|$,
	$j + w \leq j + |\mathit{Tail}_{k}| + 2^{k} < m - |\mathit{Tail}_{k+1}| + |\mathit{Tail}_{k}| + 2^{k}$.
	Since $|Tail_{k+1}| - |\mathit{Tail}_{k}| \geq 2^{k}$,
	$m - |\mathit{Tail}_{k+1}| + |\mathit{Tail}_{k}| + 2^{k} < m$.
	Thus, $j + w \leq m$.
	
	If $a=p_k$, $w=W[p_k]$ is the tight witness for offset $p_k$, i.e., $w = \LCP_P(p_k) + 1$. 
	Since $|\Tail_{k+1}| \geq \LCP_P(p_k)$,
	$j + w \leq j + |\Tail_{k+1}| + 1$. 
	Since $j < |\Head_{k+1}|$, $j + |\Tail_{k+1}| + 1 \leq |\Head_{k+1}| + |\Tail_{k+1}| \leq m$.
	We have proved that $j + w \leq m$. 
\end{proof}
Algorithm~\ref{alg:satisfy_sparsity} updates the witness table in accordance with Lemma~\ref{lem:headwitnesses}.
In this way, the $2^k$-sparsity of the head and the correctness of (non-zero) witnesses in the head are maintained.
The invariants $W[i] \leq |\Tail_k| + 2^k$ and $|\Head_k| + \LCP_P(p_k) \geq m$ are used in the proof of Lemma~\ref{lem:witness_limit}.
It remains to show the invariant.
\begin{restatable}{lemma}{lastinvariant}\label{lem:smallwitness}
	At the beginning of round $k$, 
	for all $i \in \{0,\dots,2^{k}-1\}$, it holds $W[i] \leq |\Tail_k| + 1$ and
	for all $i \in \{2^{k},\dots,|\Head_k|-1\}$, it holds $W[i] \leq |\Tail_k| + 2^k$.
\end{restatable}
\begin{proof}
	We show the lemma by induction on $k$.
	At the beginning of round $0$, every element of $W$ is zero and $|\mathit{Tail}_0| = 0$, thus, the claim holds.
	We will show that the lemma holds for $k+1$ assuming that it is the case for $k$.

	Suppose $i < 2^{k+1}$ and $i \neq p_k$. Then $W[i]$ is not updated.
	By induction hypothesis, $W[i] \leq |\Tail_k| + 2^k \leq |\Tail_{k+1}|$ holds.
	Suppose $i = p_k$. If $\mcal{W}_P(p_k)=\emptyset$, the algorithm sets $W[p_k]=0$ and thus the claim holds.
	If $\mcal{W}_P(p_k) \neq \emptyset$, the algorithm sets $W[p_k]$ to the tight witness $\LCP_P(p_k)+1$.
	Thus, $W[p_k] = \LCP_P(p_k)+1 \leq |\Tail_{k+1}|+1$.
	
	Suppose $2^{k+1} \le i < |\Head_{k+1}|$.
	If Algorithm~\ref{alg:satisfy_sparsity} does not update $W[i]$, by the induction hypothesis, $W[i] \leq |\Tail_k|+2^k < |\Tail_{k+1}|+2^{k+1}$ holds.
	Suppose Algorithm~\ref{alg:satisfy_sparsity} updates $W[i]$ or $W[j]$ by a duel between $i$ and $j$, where $2^{k+1} \leq i < j < |\Head_{k+1}|$ and $a=j-i < 2^{k+1}$.
	We have shown above that $W[a] \leq |\Tail_{k+1}|+1$.
	If $i$ wins the duel, then $W[j] = W[a] \leq |\Tail_{k+1}| + 1 \leq |\Tail_{k+1}| + 2^{k+1}$.
	If $j$ wins the duel, then $W[i] = W[a] + a \leq |\Tail_{k+1}| + 1 + a \leq |\Tail_{k+1}| + 2^{k+1}$.
\end{proof}

\begin{restatable}{lemma}{headcomplexity}
    \label{lem:satisfy_sparsity}
	In the round $k$ of the while loop, Algorithm~\ref{alg:satisfy_sparsity} updates the witness table so that $\Head_{k+1}$ is $2^{k+1}$-sparse in $O(\xi_m^\mrm{t})$ time and $O(\xi_m^\mrm{w} \cdot m/2^k)$ work on P-CRCW PRAM.
\end{restatable}
\begin{proof}
	Before the execution of Algorithm~\ref{alg:satisfy_sparsity}, since the preprocessing invariant is satisfied and $W[p_k]$ holds the tight witness for offset $p_k$, $W[0 \mathbin{:} 2^k - 1]$ contains one zero. 
	Now, let us consider a $2^{k}$-block $B$ of $\Head_k$ that is not the first $2^k$-block.
	Since $\Head_k$ satisfies the $2^{k-1}$-sparsity, there are at most two zero positions in $B$.
	Suppose that $B$ has two distinct zero positions $i$ and $i+a$.
	Since $a < 2^{k}$, $W[a] \neq 0$.
	By Lemma~\ref{lem:witness_limit}, for offsets $i$ and $i+a$, $i + a \leq m - W[a]$.
	Thus, by Lemma~\ref{lem:headwitnesses}, at least one of $W[i]$ and $W[i+a]$ is updated as the result
	of the duel.
	Thus, after performing duels for all $2^k$-blocks of $\Head_k$, $\Head_k$ satisfies the $2^{k}$-sparsity.
	
	Since each duel takes $O(\xi_m^\mrm{t})$ time and $O(\xi_m^\mrm{w})$ work and there are $O(m/2^k)$ duels in total,
	the overall time and work complexities are $O(\xi_m^\mrm{t})$ and $O(\xi_m^\mrm{w} \cdot m/2^k)$, respectively.
\end{proof}

\subsubsection*{Tail invariant}

Next, we discuss how the algorithm finalizes $\mathit{Tail}_{k+1}$ in the round $k$.
This procedure is described in Algorithm~\ref{alg:extend_suffix}.
For the sake of convenience, we denote by $\mathcal{T}_{k}$ the set of positions of $\Tail_{k}$.
Since $\Tail_{k}$ has already been finalized, it is enough to update $W[i]$ for $i \in \mcal{T}_{k+1} \setminus \mcal{T}_{k}$.
We have two cases depending on how much the tail is extended. 

The first case where $|\Tail_{k+1}| = |\Tail_{k}| + 2^{k}$ is handled naively.
Since $\Head_{k}$ satisfies the $2^{k}$-sparsity by the invariant, there are at most two zero positions in $\mathcal{T}_{k+1} \setminus \mathcal{T}_{k}$.
Algorithm~\ref{alg:extend_suffix} naively uses Algorithm~\ref{alg:check_SCER_parallel} to finalize those positions.

Now, we consider the case $|\Tail_{k+1}| = \LCP_P(p_k) > |\Tail_{k}| + 2^{k}$.
The following lemma holds concerning the periodicity under SCERs.
\begin{lemma}
\label{lem:border_period_transitivity}
    Suppose that $a$ and $b$ are periods of $X$.
    If\/ $a+b < |X|$, then $(b+a)$ is a period of $X$.
	If\/ $a < b$, then $(b-a)$ is a period of $X[1 \mathbin{:} |X| - a]$.
\end{lemma}
This lemma implies that if $p$ is a period of $X$, then so is $qp$ for every positive integer $q \le \floor{(|X|-1)/p}$.
\begin{proof}
	Let $n = |X|$.
	Since $a$ is a period of $X$, by the definition $X[1 \mathbin{:} n - a] \approx X[1+a \mathbin{:} n]$. Thus, $X[1+b \mathbin{:} n - a] \approx X[a+b+1:n]$.
	Similarly, since $b$ is a period of $X$, by the definition $X[1 \mathbin{:} n - b] \approx X[1+b \mathbin{:} n]$. Thus, $X[1 \mathbin{:} n - b - a] \approx X[1+b:n-a]$.
	Thus, $X[1+b:n-a] \approx X[a+b+1 \mathbin{:} n] \approx X[1 \mathbin{:} n - b - a]$, which means that $(b+a)$ is a period of $X$.
	
	Since $a$ and $b$ are periods of $X$, $X[1+b-a \mathbin{:} n - a] \approx X[1+b \mathbin{:} n]$ and $X[1 \mathbin{:} n - b] \approx X[1+b \mathbin{:} n]$ hold.
    Thus, by the transitivity property, $(b-a)$ is a period of $X[1:n-a]$.
\end{proof}
\begin{figure}[t]
	\centering
	\tikzset{
    charbox/.style={draw, rectangle, minimum height= 4mm, minimum width = 5mm,
    label={[align=center,font=\large]above:#1}},
    charboxbelow/.style={draw, rectangle, minimum height= 4mm, minimum width = 5mm,
    label={[align=center,font=\large]below:#1}},
    fullbox/.style={draw, rectangle, minimum height= 4mm, minimum width = 15.5cm,
    label={[align=right,font=\large]left:#1}},
    halfbox/.style={draw, rectangle, minimum height= 4mm, minimum width = 10cm,
    label={[align=right,font=\large]left:#1}},
    blockline/.style={densely dotted, semithick},
    fittedrec/.style n args={2}{draw, fit={(#1)(#2)}, inner sep=0, outer sep=0, pattern=north east lines}
}

\begin{tikzpicture}
\node[halfbox=$X$] (X1) {};
\node[halfbox, below right = 1cm and 2cm of X1.west, anchor=west] (X2) {};
\node[halfbox, below right = 2cm and 3.5cm of X1.west, anchor=west] (X3) {};
\node[below left = 1cm and 0cm of X1.west, anchor=east, font=\large] {$X$};
\node[below left = 2cm and 0cm of X1.west, anchor=east, font=\large] {$X$};


\node[charbox=$n$, left = 0cm of X1.east] (x1_n) {};
\node[charbox=$n$, left = 0cm of X2.east] (x2_n) {};
\node[charbox=$1$, right = 0cm of X3.west] (x3_1) {};

\node[charbox=$b-a+1$, above = 1cm of x3_1, anchor=north] {};
\node[charbox=$b-a+1$, above = 2cm of x3_1, anchor=north] {};

\node[charbox=$b-a+1$, above = 1cm of x3_1, anchor=north] {};
\node[charbox=$b-a+1$, above = 2cm of x3_1, anchor=north] {};

\node[charbox=$n-a$, below = 1cm of x1_n, anchor=south] {};
\node[charbox=$n-b$, below = 2cm of x1_n, anchor=south] {};

\node[charbox=$n-b-a$, below = 1cm of x2_n, anchor=south] {};


\draw[blockline, <->] ($(X1.west) + (0, -1cm)$)  -- (X2.west) node[below, midway]{$a$};
\draw[blockline, <->] ($(X1.west) + (0, -2cm)$)  -- (X3.west) node[below, midway]{$b$};

\end{tikzpicture}
	\caption{Suppose that $a$ and $b$ are periods of $X$.
    If\/ $a+b < |X|$, then $(b+a)$ is a period of $X$.
	If\/ $a < b$, then $(b-a)$ is a period of $X[1 \mathbin{:} |X| - a]$.
	}
	\label{fig:border_period_transitivity}
\end{figure}
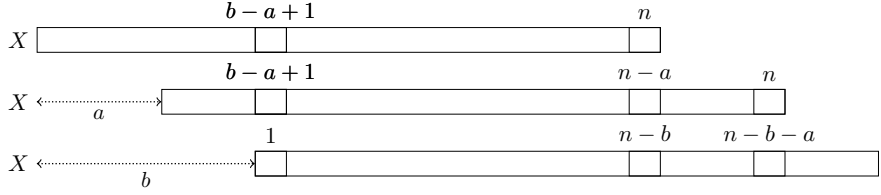
We finalize the tail based on the following lemma.
\begin{restatable}{lemma}{tailprop}
\label{lem:tail_prop}
    Suppose $m - LCP_P(p) \leq b < m$.
    If $w \in \mcal{W}_P(b)$, then $(w+b-a) \in \mcal{W}_P(a)$ for any offset $a$ such that $0 \leq a \leq b$ and $a \equiv b \pmod{p}$.
\end{restatable}
\begin{proof}
	Figure~\ref{fig:tail_prop} may help understanding the proof.
	Suppose $w \in \mcal{W}_P(b)$, i.e., $\widetilde{P}_{b+1}[w] \neq \widetilde{P}[w]$.
	Since $p$ is a period of $P[1 \mathbin{:} \LCP_P(p)]$ and $a \equiv b \pmod {p}$, by Lemma~\ref{lem:border_period_transitivity}, $(b - a)$ is also a period of $P[1 \mathbin{:} \LCP_P(p)]$, i.e., $P[1+b-a : LCP_P(p)] \approx P[1:LCP_P(p) - (b-a)]$.
	Particularly for the position $w \leq m-b \leq m-a$, we have $\widetilde{P}_{b-a+1}[w] = \widetilde{P}[w]$.
	Then, $\widetilde{P}_{b-a+1}[w] \neq \widetilde{P}_{b+1}[w]$ by the assumption (Figure~\ref{fig:tail_prop}).
	By Property (3) of the $\approx$-encoding (Definition~\ref{def:encoding}), $\widetilde{P}_{1}[b-a+w] \neq \widetilde{P}_{a+1}[b-a+w]$.
	That is, $(w+b-a) \in \mathcal{W}_P(a)$.
\end{proof}
\begin{figure}[t]
	\centering
	\tikzset{
    charbox/.style={draw, rectangle, minimum height= 4mm, minimum width = 5mm,
    label={[align=center,font=\large]above:#1}},
    charboxbelow/.style={draw, rectangle, minimum height= 4mm, minimum width = 5mm,
    label={[align=center,font=\large]below:#1}},
    fullbox/.style={draw, rectangle, minimum height= 4mm, minimum width = 15.5cm,
    label={[align=right,font=\large]left:#1}},
    halfbox/.style={draw, rectangle, minimum height= 4mm, minimum width = 10cm,
    label={[align=right,font=\large]left:#1}},
    blockline/.style={densely dotted, semithick},
    fittedrec/.style n args={2}{draw, fit={(#1)(#2)}, inner sep=0, outer sep=0, pattern=north east lines}
}

\begin{tikzpicture}
\node[halfbox=$P$] (P1) {};
\node[halfbox, below right = 1cm and 2cm of P1.west, anchor=west] (P2) {};
\node[halfbox, below right = 2cm and 3.5cm of P1.west, anchor=west] (P3) {};
\node[below left = 1cm and 0cm of P1.west, anchor=east, font=\large] {$P$};
\node[below left = 2cm and 0cm of P1.west, anchor=east, font=\large] {$P$};


\node[charbox=$1$, right = 0cm of P3.west] (c1) {};
\node[charbox=$w$, right = 4cm of P3.west] (c2) {};
\node[charbox=$m-b$, below left = 2cm and 0cm of P1.east, anchor=east] (c3) {};

\node[charbox=$b-a+1$, above = 1cm of c1, anchor = north] {};
\node[charbox=$w+b-a$, above = 1cm of c2, anchor = north] {};
\node[charbox=$m-a$, above = 1cm of c3, anchor = north] {};

\draw (P3.south west) arc(-180:0:40mm and 2mm) node[below, midway]{$\mathit{LCP}_P(p)$};

\draw[blockline, <->] ($(P1.west) + (0, -1cm)$)  -- (P2.west) node[below, midway]{$a$};
\draw[blockline, <->] ($(P1.west) + (0, -2cm)$)  -- (P3.west) node[below, midway]{$b$};
\end{tikzpicture}
	\caption{For offsets $a, b$ such that $m - LCP_P(p) \leq b < m$ and $a \equiv b \pmod p$, if $w \in \mathcal{W}_P(b)$, then $(w+b-a) \in \mathcal{W}_P(a)$.}
	\label{fig:tail_prop}
\end{figure}
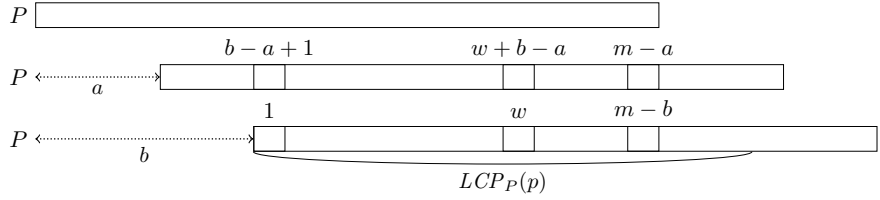
Let us partition $\mcal{T}_{k+1}\setminus \mcal{T}_{k}$ into $p_k$ subsets $\mcal{S}_0,\dots,\mcal{S}_{p_k-1}$ where $\mcal{S}_\rem = \{\, i \in \mathcal{T}_{k+1} \setminus \mcal{T}_{k} \mid i \equiv \rem \pmod{p_k}\,\}$,
some of which can be empty.
Lemma~\ref{lem:tail_prop} implies that for each $\rem \in \{0,\dots,p_k-1\}$, there exists a boundary offset $b_{\rem}$ such that, for every $i \in \mathcal{S}_\rem$, $\mathcal{W}_P(i) = \emptyset$ iff $i > b_{\rem}$.
Fortunately, for many $\rem$, one can find the boundary $b_\rem$ very easily, unless $\mcal{S}_\rem=\emptyset$.
Let $q_\rem = \max \mcal{S}_\rem$ for non-empty $\mcal{S}_\rem$. 
Due to the $2^{k}$-sparsity and the fact $p_k < 2^{k+1}$, it holds $W[q_\rem] \neq 0$ for all but at most three $\rem$.
If $W[q_\rem] \neq 0$, then $q_\rem$ is the boundary.
By Lemma~\ref{lem:tail_prop}, $W[W[q_\rem] + q_\rem-i] \in \mcal{W}_P(i)$ for all $i \in \mcal{S}_\rem$.
Accordingly, Algorithm~\ref{alg:extend_suffix} updates those values $W[i]$ in parallel in Lines~\ref{alg:line:suffix_1}--\ref{alg:line:suffix_2}.

\begin{algorithm2e}[tb]
	\caption{Finalize $\Tail_{k+1}$.}
	\label{alg:extend_suffix}
    \SetVlineSkip{0.5mm}
    \Fn{\FinalizeTail{$\tail, \oldtail, p, k$}}{
		\If{$\oldtail - \tail = 2^k$}{
			$Z \leftarrow \GetZeros{$\mathit{\tail}, \oldtail - 1,k$}$\tcc*{$|Z|\leq 2$}
	        \For{$i = 0$ \KwTo $|Z| - 1$}{
    	        $z \leftarrow Z[i]$\;
        	    \If{$z \neq -1$}{%
            	     $W[z] \leftarrow \CheckParallel{$\widetilde{P}[1 \mathbin{:} m - z], \widetilde{P}_{z+1}[1 \mathbin{:} m - z]$}$
                }
            }
		}\Else{
        \ForPar{\textbf{each} $i \in \{\tail, \dotsc, \oldtail-1\}$}{
            \label{alg:line:suffix_1}
            $q \leftarrow  j$ where $j \in \{\oldtail-p, \dotsc, \oldtail-1\}$ and $j \equiv i \pmod p$\;
            \lIf{$W[i] = 0$ and $W[q] \neq 0$}{$W[i] \Leftarrow W[q] + q-i$}
            \label{alg:line:suffix_2}
        }
        $Z \leftarrow $ \GetZeros{$\mathit{old\_tail} - p, \oldtail - 1,k$}\tcc*{$|Z|\leq 3$}
        \For{$i = 0$ \KwTo $|Z| - 1$}{
            $z \leftarrow Z[i]$\;
            \lIf{$z \neq -1$}{%
               \Finalize{$\tail, \oldtail, p, z \bmod p$}%
                }
    	    }
	    }
    }
\end{algorithm2e}
\begin{algorithm2e}[tb]
	\caption{Finalize $i \in \mathcal{T}_{k+1} \backslash \mathcal{T}_{k}$ s.t. $i \equiv \rem \pmod {p_k}$.}
	\label{alg:binary_search}
    \SetVlineSkip{0.5mm}
    \Fn{\Finalize{$\tail, \oldtail, p, \rem$}}{
        $l \leftarrow \ceil{(\tail-\rem)/p} - 1, \ r \leftarrow \floor{(\oldtail-1-\rem)/p} + 1$\;
        \While{$r - l > 1$}{
            $i \leftarrow \floor{(l + r)/2},\ j \leftarrow i \cdot p + \rem$\;
            \lIf{\CheckParallel{$\widetilde{P}[1 \mathbin{:} m - j],\widetilde{P}_{j+1}[1 \mathbin{:} m-j]$} $ = 0$}{%
                \label{alg:line:binary_search}%
                $r \leftarrow i$
            }
            \lElse{$l \leftarrow i$}
        }
        $b_{\rem} \leftarrow l \cdot p + \rem$\;
        $w \leftarrow $ \CheckParallel{$\widetilde{P}[1 \mathbin{:} m - b_{\rem}], \widetilde{P}_{b_{\rem}+1}[1 \mathbin{:} m - b_{\rem}]$}\;
        \ForPar{\textbf{each} $i \in \{\mathit{tail}, \dotsc, b_{\rem}\}$}{
            \lIf{$W[i] = 0$ and $i \equiv b_{\rem} \pmod p$}{%
                \label{alg:line:binary_search_update}%
                $W[i] \Leftarrow w + b_{\rem}-i$%
            }
        }
    }
\end{algorithm2e}

On the other hand, for $\rem$ such that $W[q_\rem] = 0$, Algorithm~\ref{alg:binary_search} uses binary search to find $b_{\rem}$ and a witness $w \in \mcal{W}_P(b_\rem)$ if it exists.
Then, following Lemma~\ref{lem:tail_prop}, Algorithm~\ref{alg:binary_search} sets in parallel $W[i]$ to $w + (b_{\rem} -i)$ where $w \in \mathcal{W}_P(b_{\rem})$ for $i \in \mathcal{S}_\rem$ such that $i \leq b_{\rem}$ (Line~\ref{alg:line:binary_search_update}).
If there is no boundary $b_\rem$ in $\mcal{S}_\rem$, then $\mathcal{W}_P(i) = \emptyset$ for all $i \in \mathcal{S}_\rem$.
We do nothing in that case.

In Algorithm~\ref{alg:binary_search}, the invariant is as follows.
For $i \in \mathcal{S}_\rem$, $\mathcal{W}_P(i) \neq \emptyset$ if $i \leq l \cdot p_k + \rem$, and $\mathcal{W}_P(i) = \emptyset$ if $i \geq r \cdot p_k + \rem$.
Each condition check of the binary search (Line~\ref{alg:line:binary_search})
takes $O(\xi_m^\mrm{t})$ time and $O(\xi_m^\mrm{w} \cdot m)$ work.
Thus, the overall complexity of Algorithm~\ref{alg:binary_search} is $O(\xi_m^\mrm{t} \log m)$ time and $O(\xi_m^\mrm{w} \cdot m \log m)$ work.

\begin{restatable}{lemma}{tailcomplexity}
\label{lem:extend_tail}
	In round $k$, Algorithm~\ref{alg:extend_suffix} finalizes $\mathit{Tail}_{k+1}$
	in $O(\xi_m^\mrm{t} \log m)$ time and $O(\xi_m^\mrm{w} \cdot m \log m)$ work on P-CRCW PRAM.
\end{restatable}
\begin{proof}
	First, if $|\Tail_{k+1}| = |\Tail_{k}| + 2^{k}$, Algorithm~\ref{alg:extend_suffix} finalizes all positions $i \in \mathcal{T}_{k+1} \setminus \mathcal{T}_{k}$ in $O(1)$ time and $O(m)$ work.
	Next, let us consider the case when $|\Tail_{k+1}| = \LCP_P(p_k)$.
	Algorithm~\ref{alg:extend_suffix} finalizes all positions $i \in \mathcal{S}_\rem$ such that $W[q_\rem] \neq 0$ in $O(1)$ time and $O(m)$ work.
	Considering $i \in \mathcal{S}_\rem$ such that $W[q_\rem] = 0$,
	since $\Head_{k}$ is $2^{k}$-sparse and $2^{k} \leq p_k < 2^{k+1}$,
	there are at most three zero positions in the suffix of length $p_k$ of $\Head_{k}$.
	Therefore, there are at most three $\rem$ where $W[q_\rem] = 0$.
	Algorithm~\ref{alg:extend_suffix} updates all positions of $\mathcal{S}_\rem$ in parallel in
	$O(\xi_m^\mrm{t} \log m)$ time and $O(\xi_m^\mrm{w} \cdot m \log m)$ work.
	Thus, overall Algorithm~\ref{alg:extend_suffix} runs in $O(\xi_m^\mrm{t} \log m)$ time and $O(\xi_m^\mrm{w} \cdot m \log m)$ work.
\end{proof}

\subsection{Pattern searching}
Now we suppose that a witness table of the pattern has been computed.
Our pattern searching algorithm prunes candidates in two stages: dueling and sweeping stages.
During the dueling stage, candidate positions duel with each other, until the surviving candidate positions are pairwise consistent.
During the sweeping stage, the surviving candidates from the dueling stage are further pruned so that only pattern occurrences survive.
To keep track of the surviving candidates, we introduce a Boolean array $C[1 \mathbin{:} m]$ and initialize every entry of $C$ to $\True$.
If a candidate $T_i$ gets eliminated, we set $C[i] = \False$.
The pattern searching algorithm updates $C$ in such a way that $C[i] = \True$ iff $i$ is a pattern occurrence. 
Entries of $C$ are updated at most once during the dueling and sweeping stages.

\paragraph*{Comparison with Vishkin's algorithm}
When considering exact matching, Vishkin~\cite{vishkin1985optimal} found that if the pattern is periodic, i.e., $P=Q^{j}Q'$ for some aperiodic string $Q$, a proper prefix $Q'$ of $Q$, and $j \ge 2$, the problem can be reduced to finding occurrences of $Q$ and $Q'$ in the text.
Then a position $i$ is an occurrence of $P$ if and only if $i$ is a starting position of $j$ consecutive occurrences of $Q$ followed by an occurrence of $Q'$.
His dueling stage keeps the table $C$ to be $2^k$-sparse in the sense that $C[i]=\True$ for at most one position $i$ in every $2^k$-block, incrementing $k$ up to $\floor{\log |Q|/2}$.
This can be done without violating the invariant, since the occurrences of an aperiodic string $Q$ are guaranteed to be sparse in the sense that the distance of two consecutive occurrences is bigger than $|Q|/2$.
Then the sweeping stage naively checks whether those sparse surviving positions $i$ with $C[i]=\True$ are real occurrences.
Apparently, this idea does not work in SCER matching.
If $P$ has a period $p$ under an SCER, it does not mean that $P$ is a repetition of $Q=P[1:p]$ or that consecutive occurrences of $Q$ form an occurrence of $P$.
Our dueling and sweeping algorithms presented here are quite different from Vishkin's.

\subsubsection*{Dueling stage}

The dueling stage is described in Algorithm~\ref{alg:dueling_stage_general}.
A set of positions is said to be \emph{consistent} if all elements in the set are pairwise consistent.
During the round $k$, the algorithm partitions the candidate positions into blocks of size $2^k$.
Let $\mathcal{C}_{k,j} \subseteq \{(j-1)2^k+1,\dots,j \cdot 2^k\}$ be the set of candidate positions in the $j$-th $2^k$-block which have survived after the round $k$.
The invariant of Algorithm~\ref{alg:dueling_stage_general} is as follows.
\begin{itemize}
\item At any point of execution of Algorithm~\ref{alg:dueling_stage_general}, all pattern occurrences survive.
\item For round $k$, each $\mathcal{C}_{k,j}$ is consistent.
\end{itemize}
Set $\mathcal{C}_{k,j}$ is obtained by ``merging'' $\mathcal{C}_{k-1,2j-1}$ and $\mathcal{C}_{k-1,2j}$.
That is, $\mathcal{C}_{k,j}$ shall be a consistent subset of $\mathcal{C}_{k-1,2j-1} \cup \mathcal{C}_{k-1,2j}$ which contains all the occurrence positions in $\mathcal{C}_{k-1,2j-1} \cup \mathcal{C}_{k-1,2j}$.
After the dueling stage, $\mcal{C}_{\ceil{\log m},1}$ is a consistent set including all the occurrence positions.
We then let $C[i]=\True$ iff $i \in \mcal{C}_{\ceil{\log m},1}$.
In our algorithm, each set $\mathcal{C}_{k,j}$ is represented as an integer array, where elements are sorted in increasing order.

\begin{algorithm2e}[t]
	\Fn(){\DuelingStageParallel{}}{
	\ForPar{\textbf{each} $j \in \{1, \dotsc, m\}$}{
		$\mathcal{C}_{0, j}[1] \Leftarrow j$\;
	}
	$k \leftarrow 1$\;
	\While{$k \leq \ceil{\log m}$}{
		\ForPar{\textbf{each} $j \in \{1, \dotsc, \ceil{m/2^k}\}$}{
			$\mathcal{A} \leftarrow \mathcal{C}_{k-1, 2j - 1},\ \mathcal{B} \leftarrow \mathcal{C}_{k-1, 2j}$\;
			$\langle a, b \rangle \leftarrow $ \Merge{$\mathcal{A}, \mathcal{B}$}\;
			Let $\mathcal{C}_{k,j}$ be array of length $(a + |\mathcal{B}| - b + 1)$\;
			\ForPar{\textbf{each} $i \in \{1, \dotsc, a\}$}{
				$\mathcal{C}_{k,j}[i] \Leftarrow \mathcal{A}[i]$\;
			}
			\ForPar{\textbf{each} $i \in \{b, \dotsc, |\mathcal{B}|\}$}{
				$\mathcal{C}_{k,j}[a + i - b + 1] \Leftarrow \mathcal{B}[i]$\;
			}
			
		}
		$k \leftarrow k+1$\;
	}
	Initialize all elements of $C$ to $\False$\;
	\ForPar{\textbf{each} $i \in \{1, \dotsc, |\mathcal{C}_{\ceil{\log m},1}|\}$}{
		$C[\mathcal{C}_{\ceil{\log m},1}[i]] \Leftarrow \True$\;
	}
	}
	\caption{Parallel algorithm for the dueling stage.}
	\label{alg:dueling_stage_general}
\end{algorithm2e}

Let us consider merging two respectively consistent sets $\mathcal{A}(=\mathcal{C}_{k-1,2j-1})$ and $\mathcal{B}(=\mathcal{C}_{k-1,2j})$ where $\mathcal{A}$ \emph{precedes} $\mathcal{B}$, i.e., $\max \mathcal{A} < \min \mathcal{B}$.
Sets $\mathcal{A}$ and $\mathcal{B}$ should be merged in such a way that the resulting set is consistent and contains all occurrences in $\mathcal{A}$ and $\mathcal{B}$.
That is, we must find a consistent set $\mcal{C}$ such that $\widehat{\mcal{A}} \cup \widehat{\mcal{B}} \subseteq \mcal{C} \subseteq \mcal{A} \cup \mcal{B}$ where $\widehat{\mcal{A}} = \{\, a \in \mcal{A} \mid T_a \approx P\,\}$ and $\widehat{\mcal{B}} = \{\, b \in \mcal{B} \mid T_b \approx P\,\}$ are the sets of occurrences in $\mcal{A}$ and $\mcal{B}$, respectively.
\begin{lemma}\label{lem:duel_consistent}
	Suppose that we are given two respectively consistent position sets $\mathcal{A}$ and $\mathcal{B}$ such that $\mathcal{A}$ precedes $\mathcal{B}$.
	If $a \in \mathcal{A}$ and $b \in \mathcal{B}$ are consistent, then $\mcal{A}_{\leq a} \cup \mcal{B}_{\geq b}$ is also consistent, where $\mcal{A}_{\leq a} = \{i \in \mathcal{A} \mid i \leq a\}$ and $\mcal{B}_{\geq b} = \{j \in \mathcal{B} \mid j \geq b\}$.
\end{lemma}
\begin{proof}
	It is enough to show that if $i$ and $j$ are consistent and $j$ and $k$ are consistent, then $i$ and $k$ are consistent for $i < j < k$.
	This claim can be rephrased so that if $j-i$ and $k-j$ are periods, then so is $k-i$.
	This is an immediate corollary to Lemma~\ref{lem:border_period_transitivity}.
\end{proof}
Therefore, it suffices to find $(a,b) \in \mcal{A} \times \mcal{B}$ such that $a \ge \max \widehat{\mcal{A}}$, $b \le \min\widehat{\mcal{B}}$, and $a$ and $b$ are consistent.
Then, $\mcal{A}_{\leq a} \cup \mcal{B}_{\geq b}$ has the desired property.

To find such a pair $(a,b)$, let us consider a grid $G$ of size $(|\mcal{A}|+2) \times (|\mcal{B}|+2)$.
Figure~\ref{fig:dueling_stage} illustrates the grid, where indices of $\mathcal{A}$ and $\mathcal{B}$ are presented along the directions of rows and columns, respectively.
For $1 \le i \le |\mcal{A}|$ and $1 \le j \le |\mcal{B}|$, $G[i][j]$ represents the result of the duel between $\mcal{A}[i]$ and $\mcal{B}[j]$ using the witness table $W$, which are the $i$-th and $j$-th smallest elements of $\mcal{A}$ and $\mcal{B}$, respectively.
We define $G[i][j]=0$ if $W[d]=0$ for $d=\mcal{B}[j]-\mcal{A}[i]$.
If $W[d] \neq 0$ and $\mcal{A}[i]$ wins the duel, then $G[i][j] = -1$.
Otherwise, $\mcal{B}[j]$ wins the duel and $G[i][j] = 1$.
For the sake of explanatory convenience, we pad grid $G$ with $-1$s along the leftmost column, with $1$s along the bottom row, and with $0$s along the upper row and rightmost column.
Specifically,
$G[i][0] = -1$ for $i \in \{0, \dotsc, |\mathcal{A}|\}$,
$G[|\mathcal{A}| + 1][j] = 1$ for $j \in \{0, \dotsc, |\mathcal{B}|\}$, 
$G[i][|\mathcal{B}| + 1] = 0$ for $i \in \{1, \dotsc, |\mathcal{A}| + 1\}$, and
$G[0][j] = 0$ for $j \in \{1, \dotsc, |\mathcal{B}| + 1\}$.
We will not compute the whole $G$, but this concept helps understanding the behavior of our algorithm.
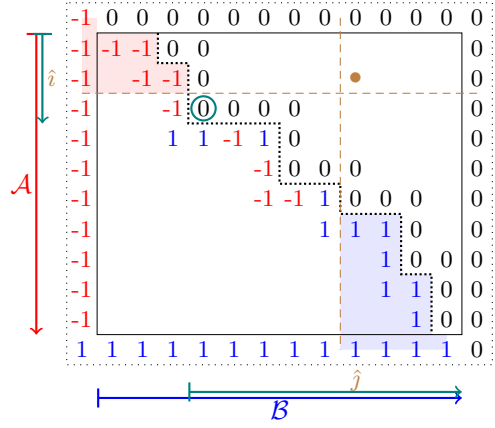
\begin{figure}[t]
 	\centering
	\newcommand{\rd}[2]{\draw (#1+0.5,#2+0.5) node {\small\textcolor{red}{-1}};}
	\newcommand{\bl}[2]{\draw (#1+0.5,#2+0.5) node {\small\textcolor{blue}{1}};}
	\newcommand{\zr}[2]{\draw (#1+0.5,#2+0.5) node {\small\textcolor{black}{0}};}
	\begin{tikzpicture}[xscale=0.4,yscale=-0.4]
		\fill[red!10] (0.5,0.5) -- (1,0.5) -- (1,1) -- (3,1) -- (3,2) -- (4,2) -- (4,3) -- (0.5,3) -- cycle;
		\fill[blue!10] (9,7) -- (11,7) -- (11,9) -- (12,9) -- (12,11) -- (12.5,11) -- (12.5,11.5) -- (9,11.5) -- cycle;
		\draw[very thin] (1,1) rectangle (13,11);
		\draw[dotted] (0,0) rectangle (14,12);
		\draw[thick, densely dotted] (3,1) -- (3,2) -- (4,2) -- (4,4) -- (7,4) -- (7,6) -- (9,6) -- (9,7) -- (11,7) -- (11,9) -- (12,9) -- (12,11);
		\draw[densely dashed,brown] (9,0.5) -- (9,11.5);
		\draw[densely dashed,brown] (0.5,3) -- (13.5,3);
		\draw (9.5,2.5) node {\textcolor{brown}{$\bullet$}};
		\draw (-0.5,2.5) node {\small\textcolor{brown}{$\hati$}};
		\draw (9.5,12.5) node {\small\textcolor{brown}{$\hatj$}};
		\draw[thick,teal] (4.5,3.5) circle (0.4);
		\draw[thick,|->,red] (-1,1) -- (-1,11);
		\draw (-1.5,6) node {\textcolor{red}{$\mcal{A}$}};
		\draw[thick,|->,blue] (1,13.1) -- (13,13.1);
		\draw (7,13.5) node {\textcolor{blue}{$\mcal{B}$}};
		\draw[thick,|->,teal] (-0.8,1) -- (-0.8,4);
		\draw[thick,|->,teal] (4,12.9) -- (13,12.9);
		\rd{1}{1}
		\rd{2}{1}
		\rd{2}{2}
		\rd{3}{2}
		\rd{5}{4}
		\rd{6}{5}
		\rd{6}{6}
		\rd{7}{6}
		\rd{3}{3}
		\bl{3}{4}
		\bl{4}{4}
		\bl{6}{4}
		\bl{8}{6}
		\bl{8}{7}
		\bl{9}{7}
		\bl{10}{7}
		\bl{10}{8}
		\bl{10}{9}
		\bl{11}{9}
		\bl{11}{10}
		\zr{3}{1}
		\zr{4}{1}
		\zr{4}{2}
		\zr{4}{3}
		\zr{5}{3}
		\zr{6}{3}
		\zr{7}{3}
		\zr{7}{4}
		\zr{7}{5}
		\zr{8}{5}
		\zr{9}{5}
		\zr{9}{6}
		\zr{10}{6}
		\zr{11}{6}
		\zr{11}{7}
		\zr{11}{8}
		\zr{12}{8}
		\zr{12}{9}
		\zr{12}{10}
		\foreach \y in {0,...,10}{\rd{0}{\y}\zr{13}{\y+1}}
		\foreach \x in {0,...,12}{\bl{\x}{11}\zr{\x+1}{0}}
	\end{tikzpicture}
	\caption{Padded grid $G$ given two consistent sets $\mathcal{A}$ and $\mathcal{B}$.
	The grid is separated into the zero region and the non-zero region by the dotted boundary line (Lemma~\ref{lem:duel_consistent}).
	The coordinate $(\hati,\hatj)$ is indicated by the brown dot.
	The red- and blue-shaded areas consist of $-1$ and $1$ only, respectively (Lemma~\ref{lem:region_of_occurrences}).
	Our algorithm outputs $(i,j)$ such that $G[i][j]=0$, $G[i][j-1]=-1$, $G[i+1][j']=1$, and $G[i+1][j'+1]=0$ for some $j'$.
	If there are more than one such coordinate, the smallest $i$ will be chosen by the priority.
	The output coordinate is indicated by the green circle above.
	Then the obtained set consists of the elements represented by the two green lines.    
	}
	\label{fig:dueling_stage}
\end{figure}

In terms of the grid representation, our goal is to find a coordinate $(i,j)$ such that $G[i][j]=0$ and it is to the lower left of $(\hati,\hatj)$ (brown dot in Figure~\ref{fig:dueling_stage}) where $\hati = \max(\{\, i' \mid \mcal{A}[i'] \in \widehat{\mcal{A}} \,\}\cup\{0\})$ and $\hatj = \min(\{\, j' \mid \mcal{B}[j'] \in \widehat{\mcal{B}} \,\} \cup \{|\mcal{B}|+1\})$.
Then, $\mcal{A}_{\le \mcal{A}[i]} \cup \mcal{B}_{\ge \mcal{B}[j]}$ has the desired property, where $\mcal{A}_{\le \mcal{A}[0]}$ and $\mcal{B}_{\ge \mcal{B}[|\mcal{B}|+1]}$ are assumed to be empty.

Lemma~\ref{lem:duel_consistent} implies that if $G[i][j]=0$ then $G[i'][j']=0$ for any $i' \le i$ and $j' \ge j$.
Therefore, grid $G$ can be divided into two regions: the upper-right region that consists of only $0$ and the rest that consists of a mixture of $-1$ and $1$.
The boundary line looks like a step function.
The distributions of $1$ and $-1$ in the non-zero region are not totally random.
Since occurrences will never lose the duel, if $\mcal{A}[i] \in \widehat{\mcal{A}}$, then row $i$ consists of non-positive elements only,
and if $\mcal{B}[j] \in \widehat{\mcal{B}}$, then column $j$ consists of non-negative elements only.
Particularly, $G[\hati][\hatj]=0$.
The following lemma strengthens this observation.
\begin{restatable}{lemma}{duelingregion}
\label{lem:region_of_occurrences}
    If $\mcal{A}[i] \leq \max\widehat{\mcal{A}}$, then row $i$ consists only of non-positive elements.
	Similarly, if $\mcal{B}[j] \leq \min\widehat{\mcal{B}}$, then column $j$ consists only of non-negative elements.
\end{restatable}
\begin{proof}
	We prove the first half of the lemma.
	The second claim can be proven in the same way.
	We show that if $i \leq \hati$ and $G[i][j] \neq 0$, then $G[i][j] = -1$ for any $1 \le j \le |\mcal{B}|$.
	Let $a=\mcal{A}[i]$, $\hat{a}=\max\widehat{\mcal{A}}$, $b=\mcal{B}[j]$, and $\hat{b}=\min\widehat{\mcal{B}}$,
	and suppose the inconsistency of $a$ and $b$ is witnessed by $W[b-a]=w \neq 0$, i.e., $\widetilde{P}[w] \neq \widetilde{P}_{b-a+1}[w]$.
	Since $T_{\hat{a}} \approx P$, $T[b:m+\hat{a}-1] \approx P[b-\hat{a}+1:m]$, which implies $\widetilde{T}_b[w] = \widetilde{P}_{b-\hat{a}+1}[w]$.
	On the other hand, since $a$ and $\hat{a}$ are consistent, i.e., $P[1:m-(\hat{a}-a)] \approx P[\hat{a}-a+1:m]$,
	we have $P[b-\hat{a}+1:m-(\hat{a}-a)] \approx P[b-a+1:m]$, which implies $\widetilde{P}_{b-\hat{a}+1}[w] = \widetilde{P}_{b-a+1}[w]$.
	Therefore, $\widetilde{T}_b[w] = \widetilde{P}_{b-\hat{a}+1}[w] = \widetilde{P}_{b-a+1}[w] \neq \widetilde{P}[w]$.
	Hence, $a$ wins the duel against $b$ and thus $G[i][j] = -1$.
\end{proof}

\begin{algorithm2e}[t]
	\Fn(){\Merge{$\mathcal{A}, \mathcal{B}$}}{
		\ForPar{\textbf{each} $i \in \{0, \dotsc, |\mcal{A}|+1\}$}{
			$j' \leftarrow 0$, $j'' = |\mcal{B}|+1$\;
			\While{$j''-j' > 1$}{
				$j \leftarrow \floor{(j'+j'')/2}$\;
				\If{$G[i][j] = 0$}{$j'' \leftarrow j$\;}
				\Else{$j' \leftarrow j$\;}
			}
			${D}[i] \Leftarrow j'$\;
		}
		\ForPar{\textbf{each} $i \in \{0, \dotsc, |\mcal{A}|\}$}{
			\If{$G[i][{D}[i]]=-1$ and $G[i+1][{D}[i+1]]=1$}{%
				$i_* \Leftarrow i$, \ $j_* \Leftarrow {D}[i]+1$\;
			}
		}
		\Return{$(i_*,j_*)$}\;
	}
	\caption{Merge two consistent sets $\mathcal{A}$ and $\mathcal{B}$}
	\label{alg:dueling_merge}
\end{algorithm2e}


Algorithm~\ref{alg:dueling_merge} firstly finds the unique column $j_i$ for each row $i$ such that $G[i][j_i] \ne 0$ and $G[i][j_i+1]=0$.
Among those boundary coordinates, the algorithm finds a neighbour pair $(i,j_i)$ and $(i+1,j_{i+1})$ such that $G[i][j_i]=-1$ and $G[i+1][j_{i+1}]=1$.
Then, it outputs $(i,j_i+1)$.

\begin{lemma}\label{lem:merge}
	Algorithm~\ref{alg:dueling_merge} finds a coordinate $(i_*,j_*)$ such that $i_* \ge \hati$, $j_* \le \hatj$, and $G[i_*][j_*]=0$ in ${O}(\xi_m^\mrm{t} \log |\mcal{B}|)$ time with ${O}(\xi_m^\mrm{w} |\mcal{A}| \log |\mcal{B}|)$ work.
\end{lemma}
\begin{proof}
	In the \textbf{while}-loop, always $G[i][j'] \ne 0$, $G[i][j''] = 0$, and $j_0 < i_1$ hold.
	When exiting the first \textbf{while}-loop, $D[i]=j'$ such that $G[i][j'] \neq 0$ and $G[i][j'+1]=0$ hold for all $i$.
	Then, the algorithm finds $i$ such that $G[i][j]=-1$, $G[i][j+1]=0$, $G[i+1][j']=1$, where $j=D[i]$ and $j'=D[i+1]$.
	Since $G[i][j]=-1$, by Lemma~\ref{lem:region_of_occurrences}, $j_i < \hatj$.
	Similarly, $G[i+1][j']=1$ implies $i+1 > \hati$.
	Thus, $i_* = i \ge \hati$ and $j_* = j+1 \le \hatj$ satisfy the desired property by Lemma~\ref{lem:duel_consistent}.
%
	
	Since a duel takes $O(\xi_m^\mrm{t})$ time and $O(\xi_m^\mrm{w})$ work, we obtain the claimed complexity.
\end{proof}

\begin{restatable}{lemma}{duelingcomplexity}\label{lem:search_dueling}
    Given a witness table, $\widetilde{P}$, and $\widetilde{T}$, Algorithm~\ref{alg:dueling_stage_general} performs the dueling stage in $O(\xi_m^\mrm{t} \log^2 m)$ time and $O(\xi_m^\mrm{w} m \log^2 m)$ work on P-CRCW-PRAM.
\end{restatable}
\begin{proof}
	Since the \textbf{while}-loop runs $O(\log m)$ times and each loop takes $O(\xi_m \log m)$ time by Lemma~\ref{lem:merge},
	the overall time complexity is $O(\xi_m \log^2 m)$.
	Now, let us look at the work complexity.
	Concerning each round $k$ of the \textbf{while}-loop of Algorithm~\ref{alg:dueling_stage_general},
	\Merge{$\mcal{A},\mcal{B}$} takes ${O}(\xi_m^\mrm{w} 2^k \log m)$ work by Lemma~\ref{lem:merge} and thus
	it takes ${O}(({m}/{2^k}) \cdot \xi_m^\mrm{w} 2^k \log m) = {O}(\xi_m^\mrm{w} m \log m)$ work.
	Since $k \in \{0, \dotsc, \ceil{\log m}\}$, the overall work complexity is ${O}(\xi_m^\mrm{w} m \log^2 m)$.
\end{proof}
\subsubsection*{Sweeping stage}
\begin{algorithm2e}[!t]
	\Fn{\SweepingStageParallel{}}{
		\textbf{create} $R[1 \mathbin{:} m]$ and initialize elements of $R$ to $0$\;
		$k \leftarrow \ceil{\log m}$\;
		\While{$k \geq 0$}{
			\textbf{create} $\mathit{Piv}[0 \mathbin{:} \floor{m/2^k}]$ and initialize its elements to $-1$\;
			\ForPar{\textbf{each} $i \in \{1, \dotsc, m\}$}{
				\lIf{$C[i] = \True$ and $(i \bmod 2^k) > 2^{k-1}$}{%
					$\mathit{Piv}[\floor{i/2^k}] \Leftarrow i$%
				}
			}
			\ForPar{\textbf{each} $b \in \{0, \dotsc, \floor{m/2^k}\}$}{
				$x \leftarrow \mathit{Piv}[b]$\;
				\If{$x \neq -1$}{%
				    $w \leftarrow $ \CheckParallel{$\widetilde{P}[R[x] + 1:m], \widetilde{T}_x[R[x] + 1:m]$\label{alg:line:sweep_0}}\;
				    \lIf{$w = 0$}{$R[x] \Leftarrow m$}
				    \lElse{$R[x] \Leftarrow R[x] + w - 1$}
					\label{alg:line:sweep_1}
				} 
			}
			\ForPar{\textbf{each} $i \in \{1, \dotsc, m\}$}{
				$x \leftarrow \mathit{Piv}[\floor{i/2^k}]$\;
				\label{alg:line:sweep_2}
				\lIf{$i \leq x$ and $R[x] \leq m-(x-i)-1$}{%
						$C[i] \Leftarrow \False$%
				}
				\label{alg:line:sweep_3}
				\lIf{$i > x$ and $C[i] = \True$}{$R[i] \Leftarrow R[x] - (i - x)$%
				\label{alg:line:sweep_4}}
			}
			$k \leftarrow k - 1$\;
		}
	}
	\caption{Parallel algorithm for the sweeping stage}
	\label{alg:sweeping_stage_general}
\end{algorithm2e}
The sweeping stage is described in Algorithm~\ref{alg:sweeping_stage_general}.
The sweeping stage updates $C$ until $C[i] = \True$ iff $i$ is a pattern occurrence.
All entries in $C$ are updated at most once.
Recall that all candidates that survived from the dueling stage are pairwise consistent.
In addition to $C$, we will create a new integer array $R[1:m]$.
Throughout the sweeping stage, we have the following invariant properties:
\begin{itemize}
    \item if $C[x]=\False$, then $T_x \not\approx P$, 
    \item if $C[x]=\True$, then $\LCP(T_x,P) \geq R[x]$.
\end{itemize}
The purpose of bookkeeping this information in $R$ is to ensure that the sweeping stage algorithm uses $O(n)$ processors in each round.
We do not want to access the same position of the text for each candidate covering the position.
For two consistent candidate positions $x$ and $x+a$ with $a > 0$, once we have calculated the value $r=\LCP(T_x,P)$, we know that $\LCP(T_{x+a},P) \ge r-a$ for free, i.e., $\widetilde{T}_{x+a}[1:r-a]=\widetilde{P}[1:r-a]$. Then it suffices to check $\widetilde{T}_{x+a}[r-a+1:m]=\widetilde{P}[r-a+1:m]$. We keep the value $r-a$ in $R[x+a]$ for this trick, if $r-a \ge 0$.
Throughout this section, we assume that a processor is attached to each position of $C$ and $T$.

For each stage $k$, $C$ is divided into $2^k$-blocks.
Unlike the preprocessing algorithm, $k$ starts from $\ceil{\log m}$ and decreases with each round until $k = 0$.
Let us look at each round in more detail.
For the $b$-th $2^k$-block of $C$, we pick as the ``pivot'' the smallest index $x_{k,b}$ in the second half of the $2^k$-block such that $C[x_{k,b}] = \True$.
In Algorithm~\ref{alg:sweeping_stage_general}, we introduce array $\mathit{Piv}[0 \mathbin{:} \floor{m/2^k}]$
where $\mathit{Piv}[b] = x_{k,b}$.
For each $x_{k,b}$, the algorithm computes $\LCP(T_{x_{k,b}},P)$ exactly and store the value in $R[x_{k,b}]$ on Lines~\ref{alg:line:sweep_0}--\ref{alg:line:sweep_1}.
Suppose that $\LCP(T_{x_{k,b}},P) < m$, i.e., $T_{x_{k,b}} \not\approx P$ and $w=\LCP(T_{x_{k,b}},P)+1$ is the tight mismatch position.
Since all surviving candidate positions are pairwise consistent, if $T_{x_{k,b}} \not\approx P$, then, any candidate $T_{x_{k,b}-a}$ that ``covers'' $w$ cannot match the pattern.
Generally, we have the following.
\begin{restatable}{lemma}{sweepingconsistent}
\label{lem:consistent_candidate}
If two candidate positions $x$ and $(x-a)$ with $a > 0$ are consistent and $LCP(T_x,P) \leq m-a-1$,
then $(x-a)$ is not an occurrence.
\end{restatable}
\begin{proof}
	Let $w = LCP(T_x,P) + 1$.
	Then $w \le m-a$ and $T_x[1:m-a] \not\approx P[1:m-a]$.
	Since $x$ is consistent with $(x-a)$,
	\(
		P[a+1:m] \approx P[1:m-a] \not\approx T_x[1:m-a] \approx T_{x-a}[a+1:m]
	\),
	which means that $(x-a)$ is not a pattern occurrence.
\end{proof}
Based on Lemma~\ref{lem:consistent_candidate}, Algorithm~\ref{alg:sweeping_stage_general} updates $C[i]$ for indices $i$ in the first half of each $2^k$-block at Line~\ref{alg:line:sweep_3}.
On the other hand, at Line~\ref{alg:line:sweep_4}, the algorithm updates the values of $R[i]$ for indices $i$ in the second half of the block if $C[i]=\True$.
Since the surviving candidates are pairwise consistent, for candidate positions $(x_{k,b}+a)$ such that $a > 0$, $T_{x_{k,b}+a}[1:r] \approx P[1:r]$ for $r = R[x_{k,b}] - a$.
In this way, the algorithm maintains the invariant properties.
When $k=0$, all the $2^k$-blocks contain just one position $x$ and $R[x]$ is set to be exactly $\LCP(T_x,P)$ by Lines~\ref{alg:line:sweep_0}--\ref{alg:line:sweep_1}, unless $C[x]=\False$ at that time.
Then, if $R[x] < m$, then $C[x]$ will be $\False$ on Line~\ref{alg:line:sweep_3}.
That is, when the algorithm halts, $C[x]=\True$ iff $T_x \approx P$.

It remains to show the efficiency of the algorithm.
\begin{lemma}\label{lem:R_monotone}
	The value of each element of $R$ is never decreased.
\end{lemma}
\begin{proof}
	Suppose that $R[i]$ is updated at round $k$ and then later at round $k'$, where $k > k'$ holds.
	Let $x$ and $x'$ be the pivots of the $2^k$- and $2^{k'}$-blocks where $i$ belongs at round $k$ and $k'$, respectively.
	It must hold that $x \le x' \le i$.
	For $y = \LCP(T_{x},P)$ and $d= x'-x$, we have $T_{x}[1:y] \approx P[1:y]$ and thus $T_{x'}[1:y-d] = T_x[1+d:y] \approx P[1+d:y]$.
	Since $x$ and $x'$ are consistent, $P[1+d:y] \approx P[1:y-d]$.
	Hence, $T_{x'}[1:y-d] \approx P[1:y-d]$, i.e., $\LCP(T_{x'},P) \ge y-d$.
	Therefore,  $\LCP(T_x,P) + x \le \LCP(T_{x'},P) + x'$ and thus $\LCP(T_x,P)-(i-x) \le \LCP(T_{x'},P)-(i-x')$ holds.
	That is, at Line~\ref{alg:line:sweep_4}, the value of $R[i]$ cannot be decreased.
\end{proof}

\begin{restatable}{lemma}{sweepingstage}
\label{lem:sweeping_processor_no_overlap}
    After the round $k$, for two surviving candidate positions $i$ and $j$ with $i < j$ that do not belong to the same $2^{k-1}$-block of $C$, $i + m \leq j + R[j]$.
\end{restatable}
\begin{proof}
	See Figure~\ref{fig:sweeping_stage}.
	Before round $\ceil{\log m}$, which can be seen as after round $\ceil{\log m}+1$, since all candidate positions belong to the same $2^{\ceil{\log m}}$-block, the statement holds (base case).
	Assuming that the statement holds after the round $(k+1)$,
	we prove that it also holds after the round $k$.
	Let $R_{k+1}$ and $R_k$ be the states of the array $R$ after the rounds $(k+1)$ and $k$, respectively.
	First, let us consider the case when surviving candidate positions $i$ and $j$ do not belong to the same $2^k$-block of $C$. 
	Obviously, $i$ and $j$ cannot belong to the same $2^{k-1}$-block.
	By the induction hypothesis, $i + m \leq j + R_{k+1}[j]$.
	Since $R_{k}[j] \geq R_{k+1}[j]$ by Lemma~\ref{lem:R_monotone}, $i + m \leq j + R_{k}[j]$.
	
	Now, let us consider the case when candidate positions $i$ and $j$ belong to the same $2^k$-block of $C$.
	During round $k$, for each $2^{k}$-block of $C$, Algorithm~\ref{alg:sweeping_stage_general} chooses as surviving candidate position $x_{k,b}$ which is the smallest index in the second half of the $2^{k}$-block.
	Thus, two surviving candidates positions $i$ and $j$ of the $b$-th $2^{k}$-block belong to different $2^{k-1}$-blocks iff $i < x_{k,b} \leq j$.
	For $T_i$ to be a surviving candidate after round $k$, it must be the case that $m + i \leq \LCP(T_{x_{k,b}}, P) + x_{k,b}$.
	For $T_j$, Algorithm~\ref{alg:sweeping_stage_general} updates $R_{k}[j]$ to $\LCP(T_{x_{k,b}}, P) - (j - x_{k,b})$.
	Substituting it into the previous inequality, we get
	$m + i \leq R_{k}[j] + (j - x_{k,b}) + x_{k,b} = R_{k}[j]+j$.
\end{proof}
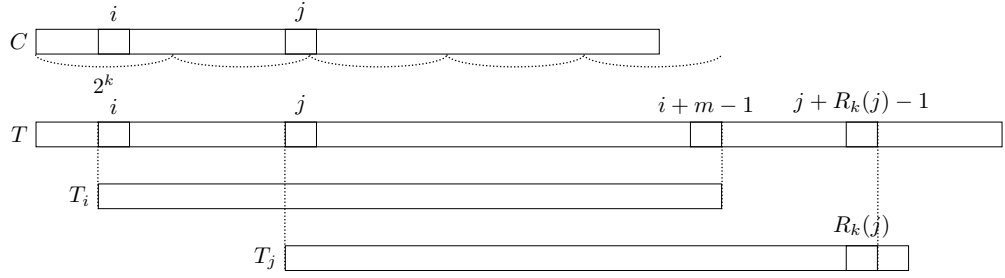
\begin{figure}[t]
	\centering
	\tikzset{
    charbox/.style={draw, rectangle, minimum height= 4mm, minimum width = 5mm,
    label={[align=center,font=\large]above:#1}},
    charboxbelow/.style={draw, rectangle, minimum height= 4mm, minimum width = 5mm,
    label={[align=center,font=\large]below:#1}},
    fullbox/.style={draw, rectangle, minimum height= 4mm, minimum width = 15.5cm,
    label={[align=right,font=\large]left:#1}},
    halfbox/.style={draw, rectangle, minimum height= 4mm, minimum width = 10cm,
    label={[align=right,font=\large]left:#1}},
    blockline/.style={densely dotted, semithick},
    fittedrec/.style n args={2}{draw, fit={(#1)(#2)}, inner sep=0, outer sep=0, pattern=north east lines}
}

\begin{tikzpicture}
\node[halfbox=$C$] (C) {};
\node[charbox=$i$, right = 1cm of C.west] (i) {};
\node[charbox=$j$, right = 4cm of C.west] (j) {};
\draw[blockline] (C.south west) arc(-180:0:11mm and 2mm) node[below, midway]{$2^k$};
\foreach \i in {1,...,4} {
    \draw[blockline] ($(C.south west) + (\i * 22mm, 0mm)$) arc(-180:0:11mm and 2mm);
}


\node[fullbox=$T$, below = 1.5cm of C.west, anchor = west] (T) {};
\node[charbox=$i$, below = 1.5cm of i, anchor = south] (ii) {};
\node[charbox=$j$, below = 1.5cm of j, anchor = south] (jj) {};

\node[charbox=$i+m-1$, right = 10.5cm of T.west] (x) {};
\node[charbox=$j+R_k(j)-1$, right = 13cm of T.west] (y) {};

\node[halfbox=$T_i$, below right = 1cm and 1cm of T.west, anchor = west] (Ti) {};
\node[halfbox=$T_j$, below right = 2cm and 4cm of T.west, anchor = west] (Tj) {};

\node[charbox=$R_k(j)$, below = 2cm of y, anchor=south] (yy) {};

\draw[blockline] (ii.north west) -- (Ti.south west);
\draw[blockline] (x.north east) -- (Ti.south east);

\draw[blockline] (jj.north west) -- (Tj.south west);
\draw[blockline] (y.north east) -- (yy.south east);





\end{tikzpicture}
	\caption{Before round $k$, for two surviving candidates $T_i$ and $T_j$ such that $j - i \geq 2^k$, $i + m - 1 < j + R_k[j]$.}
	\label{fig:sweeping_stage}
\end{figure}
\begin{restatable}{lemma}{sweepinground}
\label{lem:sweeping_round}
	Each round of the while loop of Algorithm~\ref{alg:sweeping_stage_general} can be performed in $O(\xi_m^\mrm{t})$ time with $O(n)$ processors.
\end{restatable}
\begin{proof}
	See Figure~\ref{fig:sweeping}.
	Obviously it runs in constant time except for the computation at Line~\ref{alg:line:sweep_0}, where
	each processor attached to position $i$ is used for re-encoding $\widetilde{T}[i]$ into $\widetilde{T}_x[i-x+1]$ and comparing the value with $\widetilde{P}[i-x+1]$ for some $x$.
	Indeed, there is at most one $b$ such that $x_{k,b}+R[x_{k,b}] \le i < x_{k,b}+m$,
	since $x_{k,b-1}+m \le x_{k,b}+R[x_{k,b}]$ for all $b \in \{1,\dots,\ceil{m/2^k}\}$ by Lemma~\ref{lem:sweeping_processor_no_overlap}.
\end{proof}
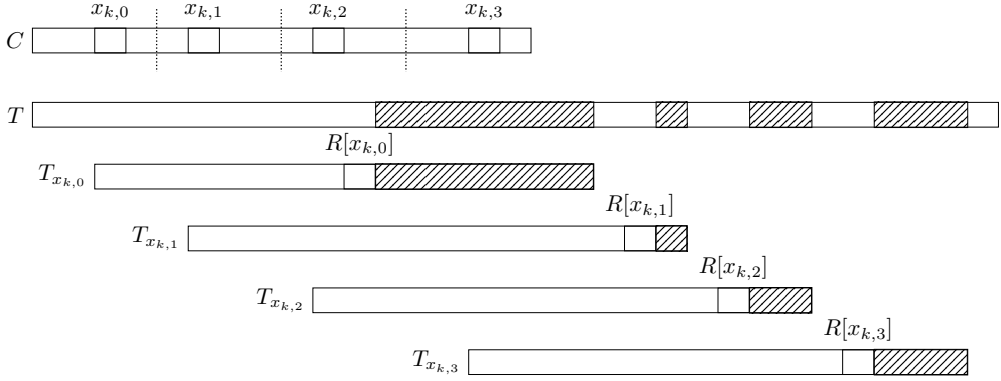
\begin{figure}[tb]
\centering
	\tikzset{
    charbox/.style={draw, rectangle, minimum height= 4mm, minimum width = 5mm,
    label={[align=center,font=\large]above:#1}},
    charboxbelow/.style={draw, rectangle, minimum height= 4mm, minimum width = 5mm,
    label={[align=center,font=\large]below:#1}},
    fullbox/.style={draw, rectangle, minimum height= 4mm, minimum width = 15.5cm,
    label={[align=right,font=\large]left:#1}},
    halfboxshort/.style={draw, rectangle, minimum height= 4mm, minimum width = 8cm,
    label={[align=right,font=\large]left:#1}},
    blockline/.style={densely dotted, semithick},
    fittedrec/.style n args={2}{draw, fit={(#1)(#2)}, inner sep=0, outer sep=0, pattern=north east lines}
}

\begin{tikzpicture}

\node[halfboxshort=$C$] (C) {};
\node[fullbox=$T$, below = 1.2cm of C.west, anchor = west] (T) {};
\draw[blockline] ($(C.north west) + (2cm, 3mm)$) -- ($(C.south west) + (2cm, -3mm)$);
\draw[blockline] ($(C.north west) + (4cm, 3mm)$) -- ($(C.south west) + (4cm, -3mm)$);
\draw[blockline] ($(C.north west) + (6cm, 3mm)$) -- ($(C.south west) + (6cm, -3mm)$);
\node[charbox=$x_{k,0}$, right = 1cm of C.west] {};
\node[charbox=$x_{k,1}$, right = 2.5cm of C.west] {};
\node[charbox=$x_{k,2}$, right = 4.5cm of C.west] {};
\node[charbox=$x_{k,3}$, right = 7cm of C.west] {};
\node[halfboxshort=$T_{x_{k,0}}$, below right = 1cm and 1cm of T.west, anchor = west] (T1) {};
\node[halfboxshort=$T_{x_{k,1}}$, below right = 2cm and 2.5cm of T.west, anchor = west] (T2) {};
\node[halfboxshort=$T_{x_{k,2}}$, below right = 3cm and 4.5cm of T.west, anchor = west] (T3) {};
\node[halfboxshort=$T_{x_{k,3}}$, below right = 4cm and 7cm of T.west, anchor = west] (T4) {};
\node[charbox=$R\lbrack x_{k,0} \rbrack$, right = 4cm of T1.west] (r1) {};
\node[charbox=$R\lbrack x_{k,1} \rbrack$, right = 7cm of T2.west] (r2) {};
\node[charbox=$R\lbrack x_{k,2} \rbrack$, right = 6.5cm of T3.west] (r3) {};
\node[charbox=$R\lbrack x_{k,3} \rbrack$, right = 6cm of T4.west] (r4) {};
\node[fittedrec={r1.north east}{T1.south east}] {};
\node[fittedrec={r2.north east}{T2.south east}] {};
\node[fittedrec={r3.north east}{T3.south east}] {};
\node[fittedrec={r4.north east}{T4.south east}] {};
\node[fittedrec={r1.north east}{T1.south east}, yshift=+1cm] {};
\node[fittedrec={r2.north east}{T2.south east}, yshift=+2cm] {};
\node[fittedrec={r3.north east}{T3.south east}, yshift=+3cm] {};
\node[fittedrec={r4.north east}{T4.south east}, yshift=+4cm] {};

\end{tikzpicture}
	\caption{Illustrating the sweeping stage. The shaded regions of the text $T$ are referenced during round $k$. Those referenced regions do not overlap.}
	\label{fig:sweeping}
\end{figure}

\begin{restatable}{lemma}{sweepingcomplexity}
    \label{lem:sweeping_complexity}
    Given $\widetilde{P}$ and $\widetilde{T}$, the sweeping stage algorithm finds all pattern occurrences in $O(\xi_m^\mrm{t} \log m)$ time and $O(\xi_m^\mrm{w} \cdot m \log m)$ work on the P-CRCW PRAM.
\end{restatable}
\begin{proof}
	The outer loop of Algorithm~\ref{alg:sweeping_stage_general} runs $O(\log m)$ times.
	By Lemma~\ref{lem:sweeping_round},
	each loop runs in $O(\xi_m^\mrm{t})$ time and $O(\xi_m^\mrm{w} \cdot m)$ processors.
	Thus, the total time is $O(\xi_m^\mrm{t} \cdot \log m)$ and total work is $O(\xi_m^\mrm{w} \cdot m \log m)$.
\end{proof}

By Theorem~\ref{th:preprocessing_complexity} and Lemmas~\ref{lem:search_dueling} and~\ref{lem:sweeping_complexity}, we obtain the main theorem.
Recall that when $n \geq 2m$, $T$ is cut into overlapping pieces of length $(2m - 1)$ and each piece is processed independently.
\begin{theorem}\label{th:search_general}
	Given a witness table, $\widetilde{P}$, and $\widetilde{T}$, the pattern searching solves the pattern searching problem under SCER in $O(\xi_m^\mrm{t} \cdot \log^2 m)$ time and $O(\xi_m^\mrm{w} \cdot n \log^2 m)$ work on the P-CRCW PRAM.
\end{theorem}
\section{Conclusion}

Dueling~\cite{vishkin1985optimal} is a powerful technique, which enables us to perform pattern matching efficiently.
In this paper, we have generalized the dueling technique for SCERs and have proposed a duel-and-sweep algorithm that solves the pattern matching problem for any SCER.
Our algorithm is the first algorithm to solve any SCER pattern matching problem in parallel.
Given a witness table, $\widetilde{P}$, and $\widetilde{T}$, we have shown that pattern searching under any SCER can be performed in $O(\xi_m^\mrm{t} \log^2 m)$ time and $O(\xi_m^\mrm{w} n \log^2 m)$ work on P-CRCW PRAM.
Given $\widetilde{P}$, a witness table can be constructed in $O(\xi_m^\mrm{t} \log^2 m)$ time and $O(\xi_m^\mrm{w} \cdot m \log^2 m)$ work on P-CRCW PRAM.
The third condition of $\approx$-encoding in \cref{def:encoding} ensures the generality of our duel-and-sweep algorithm for SCERs.
However, some standard encoding method of an SCER, namely the nearest neighbor encoding for order-preserving matching, does not fulfill the third condition.
We do not know if there is an alternative encoding for order-preserving matching that fulfills the condition and is computationally as cheap as the nearest neighbor encoding.
Nevertheless, Jargalsaikhan et al.~\cite{jargalsaikhan2018duel, jargalsaikhan2020parallel} succeeded in designing a parallel duel-and-sweep algorithm for order-preserving matching using the nearest neighbor encoding, which appears quite similar to the SCER algorithm proposed in this paper.
In our future work, we would like to investigate the relation between the encoding function and the dueling technique and further generalize the definition of encoding so that it becomes more inclusive.


\bibliography{reference}

\appendix

\section{Examples of encoding}
\label{app:sec:encoding}

\paragraph*{Prev-encoding for parameterized matching}
For a string $X$ of length $n$ over $\Sigma \cup \Pi$, where $\Pi$ is an alphabet of parameter symbols and $\Sigma$ is an alphabet of constant symbols,
the \emph{prev-encoding}~\cite{baker1996parameterized} for $X$, denoted by $\prev{X}$,
is defined to be a string over $\Sigma \cup \mathbb{N}$ of length $n$ such that for each $1 \leq i \leq n$,
\begin{align*}
	\prev{X}[i] = \begin{cases}
	X[i] & \text{if }X[i] \in \Sigma , \\
    0           & \text{if }X[i] \in \Pi \text{ and } X[i] \neq X[j] \text{ for } 1 \le j < i,\\
    i-k  		& \text{if }X[i] \in \Pi \text{ and } k = \max\{j \mid X[j]=X[i] \text{ and } 1 \le j < i\}.
  \end{cases}
\end{align*}

\begin{theorem}
Given a string $X$ of length $n$, $\prev{X}$ can be computed in $O(\log n)$ time
and $O(n \log n)$ work on P-CRCW PRAM.
Moreover, given $\prev{X}$, $\prev{X[x:n]}[i]$ can be computed in $O(1)$ time and $O(1)$ work.
\end{theorem}
\begin{proof}
    Without loss of generality, we assume that $\Pi$ forms a totally ordered domain.
    We will construct the following string $X'$ from $X$.
    We define a new symbol, say $\infty$, such that, for any element $\pi \in \Pi$, $\pi$ is less than $\infty$. 
    For $1 \leq i \leq |X|$, $X'[i] = X[i]$ if $X[i] \in \Pi$ and $X'[i] = \infty$ if $X[i] \in \Sigma$. 
    For $X'$, we construct $\mathit{Lmax}_{X'}$, which is defined as $\mathit{Lmax}_{X'}[i]=j$ if $X'[j]=\max_{k < i} \{X'[k] \mid  X'[k] \le X'[i] \}$.
    We use the rightmost (largest) $j$ if there exist more than one such $j < i$. 
    If there is no such $j$, then we define $\mathit{Lmax}_{X'}[i] = 0$.
    Suppose that $X[i] \in \Pi$ for $1 \leq i \leq |X|$.
    After computing $\mathit{Lmax}_{X'}$, $\prev{X}[i] = i - \mathit{Lmax}_{X'}[i]$ if $X[i] = X[\mathit{Lmax}_{X'}[i]]$.
    If $\mathit{Lmax}_{X'}[i] = 0$ or $X[i] \neq X[\mathit{Lmax}_{X'}[i]]$,
    then $X[i]$ is the first occurrence of this letter.
    Thus, $\prev{X}$ can be computed from $\mathit{Lmax}_{X'}$ in $O(1)$ time and
	$O(n)$ work.
    Since $\mathit{Lmax}_{X'}$ can be computed in $O(\log n)$ time and
	$O(n \log n)$ work~\cite{jargalsaikhan2020parallel}, overall complexities are $O(\log n)$ time and $O(n \log n)$ work.
	
	Given $\prev{X}$, $\prev{X[x:n]}[i]$ can be computed in the following manner in $O(1)$ time and $O(1)$ work.
    \[
       \prev{X[x:n]}[i] = \begin{cases}
	    0 & \text{if } X[x+i-1] \in \Pi \text{ and } \prev{X}[x+i-1] \geq i, \\
        \prev{X}[x+i-1]         & \text{otherwise}.
        \end{cases}
    \]
\end{proof}

\paragraph*{Parent-distance encoding for cartesian-tree matching}
For a string $X$ over a totally ordered alphabet, its parent-distance encoding~\cite{park2019cartesian} for cartesian-tree matching $\mathit{PD}_X$ is defined as follows.
\begin{align*}
	\mathit{PD}_X[i] = \begin{cases}
    0         & \text{if there is no $j < i$ such that $X[j] \leq X[i]$},
\\	i - \max_{1 \leq j < i}\{j \mid X[j] \leq X[i]\} & \text{otherwise}. 
  \end{cases}
\end{align*}

\begin{theorem}
Given a string $X$ of length $n$, $\mathit{PD}_X$ can be computed in $O(\log n)$ time and $O(n \log n)$ work on P-CRCW PRAM.
Moreover, given $\mathit{PD}_X$, $\mathit{PD}_{X[x:n]}[i]$ can be computed in $O(1)$ time and $O(1)$ work.
\end{theorem}
\begin{proof}
For $1 \leq i \leq n$, $\mathit{PD}_X[i]$ is the nearest smaller value to the left of $X[i]$.
Since the all-smaller-nearest-value problem can be solved in $O(\log n)$ time and $O(n \log n)$ work on P-CRCW PRAM by Berkman et al.~\cite{berkman1993optimal}, $\mathit{PD}_X$ can be  computed in $O(\log n)$ time and $O(n \log n)$ work on P-CRCW PRAM.

Given $\mathit{PD}_X$, $\mathit{PD}_{X[x:n]}[i]$ can be computed in the following manner in $O(1)$ time and $O(1)$ work.
\[
	\mathit{PD}_{X[x:n]}[i] = \begin{cases}
	0 & \text{if } \mathit{PD}_X[x+i-1] \geq i, \\
    \mathit{PD}_X[x+i-1]         & \text{otherwise}.
  \end{cases}
\]
\end{proof}

\end{document}